\documentclass[final,3p]{elsarticle}

\usepackage{lineno}
\usepackage[colorlinks=true,breaklinks=true,pdftex]{hyperref}
\modulolinenumbers[1]

\journal{GMP 2019}

\makeatletter
\def\ps@pprintTitle{%
 \let\@oddhead\@empty
 \let\@evenhead\@empty
 \def\@oddfoot{}%
 \let\@evenfoot\@oddfoot}
\makeatother

\biboptions{authoryear}

\usepackage{subcaption}
\usepackage{relsize}
\usepackage{bm}
\usepackage{amssymb}
\usepackage{amsmath}
\usepackage{amsthm}
\usepackage{algorithm}
\usepackage{algorithmic}
\usepackage{graphicx}
\usepackage[usenames, dvipsnames]{color}
\usepackage{longtable}
\usepackage{multirow}
\usepackage{booktabs}
\usepackage{tabularx}
\usepackage{wrapfig}
\usepackage{eqparbox}

\graphicspath{{figure/}}

\newcommand{\Bezier}{B\'{e}zier}
\newcommand{\NURB}{NURBS}
\newcommand{\Bspline}{B-spline}
\newcommand{\Reals}{\rm I\!R}
\newcommand\NewT[1]{\textcolor{black}{#1}}

\let\oldproofname=\proofname
\renewcommand{\proofname}{\rm\bf{\oldproofname}}

\newtheorem{theorem}{\sffamily Theorem}
\newtheorem{lemma}[theorem]{\sffamily Lemma}

\newtheorem{definition}[theorem]{\sffamily Definition}
\newtheorem{remark}[theorem]{\sffamily Remark}

\begin{document}
	
\begin{frontmatter}
%
% put the title of your paper here
%
\title{Volumetric Untrimming:\\ Precise decomposition of trimmed trivariates into tensor products}

% USE FOR THE REVIEW PROCESS
% \author{SRM submission id: paper1012}

%% ONLY USE FOR THE FINAL ACCEPTED PAPER SUBMISSION
\author[technion]{Fady Massarwi\corref{cor}}
\ead{fadym@cs.technion.ac.il}
\author[epfl]{Pablo Antolin}
\ead{pablo.antolin@epfl.ch}
\author[technion]{Gershon Elber}
\ead{gershon@cs.technion.ac.il}
\cortext[cor]{Corresponding author}
\address[technion]{Faculty of Computer Science, Technion-Israel Institute of Technology, Israel}
\address[epfl]{Institute of Mathematics, \'Ecole Polytechnique F\'ed\'erale de Lausanne, Switzerland}

%
% the rest is as usual
%

%%%%%%%%%%%%%%%%%%%%%%%%%%%%%%%%%%%%%%%%%%%%%%%%%%%%%%%%%%%%%%%%%%%%

\begin{abstract}
3D objects, modeled using Computer Aided Geometric Design (CAGD)
tools, are traditionally represented using a boundary representation
(B-rep), and typically use spline functions to parameterize these
boundary surfaces. However, recent development in physical analysis,
in isogeometric analysis (IGA) in specific, necessitates a volumetric
parametrization of the interior of the object. IGA is performed
directly by integrating over the spline spaces of the volumetric
spline representation of the object. Typically, tensor-product
\Bspline{} trivariates are used to parameterize the volumetric domain.

A general 3D object, that can be modeled in contemporary B-rep CAD
tools, is typically represented using trimmed \Bspline{} surfaces. In
order to capture the generality of the contemporary B-rep modeling
space, while supporting IGA needs, \cite{Massarwi2016} proposed the
use of trimmed trivariates volumetric elements. However, the use of
trimmed geometry makes the integration process more difficult since
integration over trimmed B-spline basis functions is a highly
challenging task~\cite{xu2017improved}.  In this work, we propose an
algorithm that precisely decomposes a trimmed \Bspline{} trivariate
into a set of (singular only on the boundary) tensor-product
\Bspline{} trivariates, that can be utilized to simplify the
integration process, in IGA. The trimmed \Bspline{} trivariate is
first subdivided into a set of trimmed \Bezier{} trivariates, at all
its internal knots. Then, each trimmed \Bezier{} trivariate, is
decomposed into a set of mutually exclusive tensor-product \Bspline{}
trivariates, that precisely cover the entire trimmed domain. This
process, denoted {\em untrimming}, can be performed in either the
Euclidean space or the parametric space of the trivariate. We present
examples of the algorithm on complex trimmed trivariates' based
geometry, and we demonstrate the effectiveness of the method by
applying IGA over the (untrimmed) results.
 
\end{abstract}


\end{frontmatter}

% Comment out for final accepted paper submission
% \linenumbers

%%%%%%%%%%%%%%%%%%%%%%%%%%%%%%%%%%%%%%%%%%%%%%%%%%%%%%%%%%%%%%%%%%%%

\section{Introduction} \label{sec:introSec}
Tensor product (\Bezier{}, \Bspline{} and \NURB{}) surfaces are widely used in CAGD due to their simple structure, mathematical form, and powerful geometrical properties, that make them intuitive to use. Consider the parametric representation of a tensor product \Bspline{} volume:
\begin{definition}
	A {\em \Bspline{} parametric trivariate} is a volumetric extension to
	parametric \Bspline{} curves and surfaces, with a three dimensional
	parametric space~\cite{cohengeometric}. A common representation
	of trivariates is by tensor product \Bspline{}s, as:
	\begin{equation}
	T(u,v,w)=\sum_{i=0}^{l}\sum_{j=0}^{m}\sum_{k=0}^{n}
	P_{i,j,k} B_{i,d_{u}}(u) B_{j,d_{v}}(v) B_{k,d_{w}}(w),
	\label{eqn:triv}
	\end{equation}
	where $T$ is defined over the 3D parametric domain $[U_{min},
	U_{max}) \times [V_{min}, V_{max}) \times [W_{min}, W_{max})$, 
	$P_{i,j,k} \in \Reals^q$, $q \ge 3$ are the control
	points of $T$, and $B_{i,d}$ is the $i$'th univariate \Bspline{}
	basis functions of degree $d$.
\end{definition}

3D geometric objects, generated with contemporary computer aided geometric design (CAGD) systems, are almost solely exploiting a boundary representation (B-rep), where these objects' boundaries are represented as trimmed parametric surfaces. In recent years, with recent developments of additive manufacturing and 3D printing as well as analysis, the demand for a full volumetric representation (V-rep) is increasing. Volumetric modeling of the inside of the object, as oppose 
to only its boundary (B-rep), can be used to describe different volumetric properties such as materials or stresses, in scalar, vector or tensor fields. Developments in physical analysis, isogeometric analysis (IGA) in specific~\cite{iga}, employs a parametrization of the object's volume. In IGA, the analysis is performed by integrating in the same spline spaces that describe the geometry. %

Tensor product trivariates are limited to a cuboid topology, making them difficult to use in creating general 3D objects, having an arbitrary topology. That is, the cuboid topology does not allow one to represent with ease general shapes, including holes.
Similar to trimmed surfaces, that enriches the B-rep modeling space of the set of objects that can be represented using tensor product surfaces~\cite{cohengeometric}, trimmed trivariates can be used to create far more general volumetric objects, compared to tensor products. A framework for a volumetric representation (V-rep) modeling is proposed in~\cite{Massarwi2016}, that suggests a representation for a general volumetric model (V-model) as well as V-model constructors and Boolean operations algorithms on V-models. In~\cite{Massarwi2016}, a V-model is composed of volumetric elements called V-cells, where each V-cell relates to one or more
(intersecting) trivariate(s), and is defined as following:

\begin{definition}
A V-rep cell (V-cell) is a 3-manifold that is defined over of one or more (intersecting) \Bspline{} tensor product trivariates. The sub-domain of the V-cell is delineated by trimming surfaces.
\end{definition}

In this work, we follow the definition of a V-cell in~\cite{Massarwi2016}, selecting one tensor product trivariate from the V-cell to be associate with the V-cell and parameterize it. 
Then, we define a {\bf trimmed trivariate} as following:
\begin{definition} \label{def:trimmed_tv}
Trimmed trivariate ${\cal T}$ is defined in the domain of a containing tensor product trivariate $T$ and \NewT{has} a set of (possibly trimmed) surfaces ${\cal S}$ as the trimming surfaces of $T$, forming a 2-manifold in the domain of $T$.
\end{definition}

The trimming surfaces in ${\cal S}$ can be classified into two groups. The first type of the trimming (trimmed) surfaces are (trimmed) boundary surfaces of $T$. The second type are trimming (trimmed) surfaces that are (trimmed) boundary surfaces of some other trivariates, ${\cal F}_i$, that are results of volumetric Boolean operations~\cite{Massarwi2016} between ${\cal T}$ and ${\cal F}_i$. Note that, in~\cite{Massarwi2016}, trimming (trimmed) surfaces of the second type are not defined explicitly in the parametric space of $T$, but only in the Euclidean space. Figure~\ref{fig:bez_subd_trimmed_trivar}(a) shows an example of a trimmed trivariate.

In IGA, the analysis is performed directly in a spline space employing 
\NewT{(suitable)} integration over the (trimmed) volumetric domain of the object, which herein means integration over trimmed trivariates. However,
integration over trimmed \Bspline{} basis functions is a highly challenging task~\cite{xu2017improved}. Methods to precisely integrate over the trimming domains are required, in order
to support a complete and accurate IGA over
trimmed-trivariates.  One way to achieve this goal, is by decomposing the trimmed-trivariates to tensor-products.

In this work, we present a {\em volumetric untrimming} process for trimmed trivariate:
\begin{definition}
The {\em Untrimming} of a trimmed trivariate ${\cal T}$, is a process in which ${\cal T}$ is decomposed into a set of ({\em potentially singular on the boundaries}) mutually exclusive tensor-product \Bspline{} trivariates that \NewT{precisely} cover the entire trimmed volume of ${\cal T}$.
\end{definition}

A trivariate $T$ is considered {\em potentially singular on the boundary} if the Jacobian of $T$ can vanish only at the boundaries, and $T$ is self intersection free.

The resulting set of tensor product trivariates can be utilized to simplify the integration process in IGA, as the integration will be performed over a non-trimmed \Bspline{} basis functions with full support. In the untrimming algorithm to be presented in this work, the trimmed-trivariate is first subdivided into trimmed \Bezier{} trivariates, at all its internal knots. Then, each trimmed \Bezier{} trivariate is decomposed into a set of tensor-product \Bspline{} trivariates. The untrimming algorithm can be performed either in the Euclidean space or the parametric space of the trimmed trivariate.

\NewT{It is important to remark that, on the contrary to classical finite element and IGA methods, in which the solution quality is intrinsically linked to the mesh quality, that is not the case in the present approach. In the analysis, the tensor-product trivariates created during the untrimming process are exclusively used for computing the integrals involved in Galerkin problems. Therefore, the solution discretization is not related to the untrimming result and its quality is not conditioned by the untrimming trivariates' Jacobians. We refer the interested reader to \cite{Marussig2018}, and references therein, where the numerical computation of integrals in trimmed domains using untrimming tiles is discussed.}

The rest of this document is organized as follows. Section~\ref{sec:prevSec} discusses related work. In Section~\ref{sec:algoSec}, the volumetric untrimming algorithm is introduced. In Section~\ref{sec:resultsSec}, some results of the untrimming algorithm are introduced, and in Section~\ref{sec:analysis}, an IGA application over trimmed trivariates, utilizing the presented untrimming algorithm, is presented. Finally, Section~\ref{sec:concludeSec} concludes this effort and discusses possible future work. 

\begin{figure*}
	\begin{center}
		\begin{tabular}{cccc}
			\mbox{\hspace{-0.28in}}
			\includegraphics[scale = 0.31]{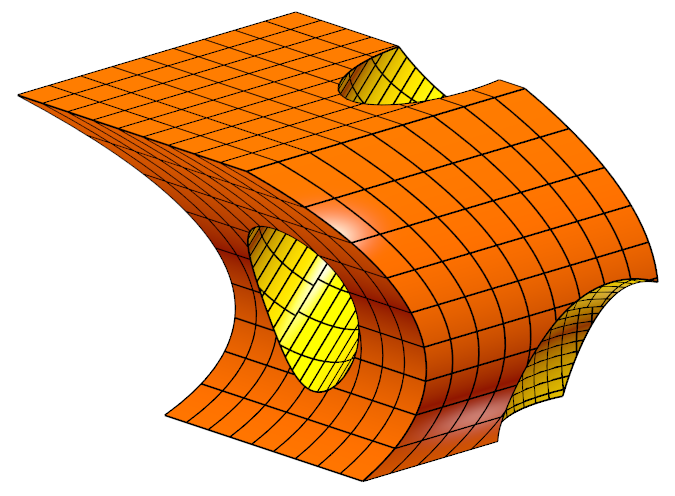} &
			\mbox{\hspace{-0.30in}}
			\includegraphics[scale = 0.31]{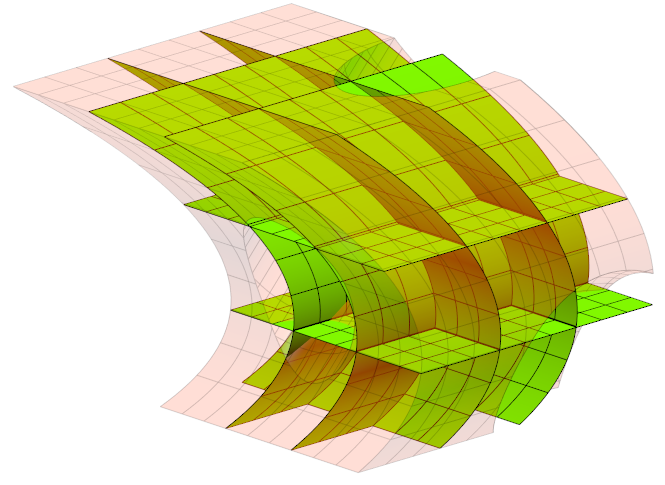} &
			\mbox{\hspace{-0.30in}}
			\includegraphics[scale = 0.31]{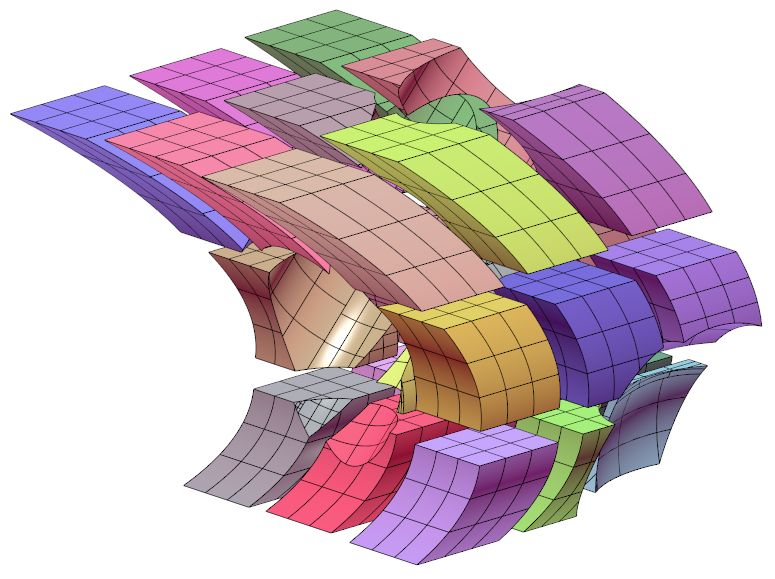}
		\end{tabular}
	\end{center}
	\begin{picture}(0,0)
	\put(120, 25){$(a)$}
	\put(275, 25){$(b)$}
	\put(440, 25){$(c)$}
	\end{picture}
	\mbox{\vspace{-0.5in}}\\[-0.5in]
	\caption{Subdivision of a trimmed \Bspline{} trivariate into trimmed \Bezier{} trivariates, shown in Euclidean space. (a) A trimmed \Bspline{} trivariate, having two internal knots at each parametric direction. Trimmed trimming surfaces on the boundaries of the trivariate have an orange color, and non-boundary trimmed trimming surfaces are colored with yellow. (b) The iso parametric surfaces of the internal knots (in green). (c) Twenty eight (and not twenty seven!) trimmed \Bezier{} trivariates in an exploded view, as result of all the subdivisions along the iso parametric knot surfaces (in (b)).}
	\label{fig:bez_subd_trimmed_trivar}
\end{figure*}
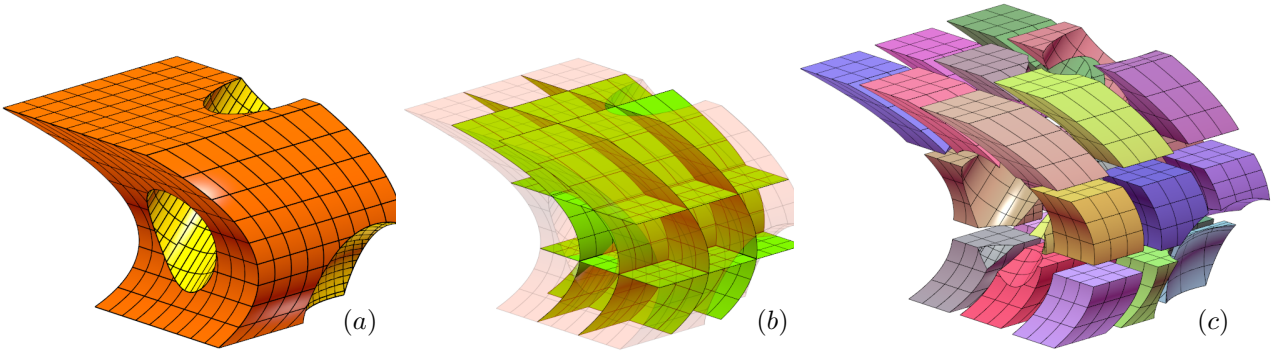

\section{Previous work}	\label{sec:prevSec}

Many algorithms have been proposed for hexahedral meshing of
2-manifold polygonal meshes,
i.e~\cite{hex_armstrong2015common,hex_yu2015recent}. Given a trimmed
trivariate, hexahedral meshing algorithms can be applied to a
polygonal mesh which is an approximation mesh of the boundary surfaces
of the trimmed trivariate, only to generate hexahedral elements that
cover the trimmed domain. Then, a trilinear tensor product trivariate
can be extracted from each hexahedron.  However, there are two
conceptual drawbacks of such approach. First, the boundaries of the
trimmed trivariate are approximated, which introduces an error which
might lead to instabilities in the integration and analysis
processes. Second, all the generated tensor product trivariates are of
linear order, where in IGA, basis functions of a higher order (and
continuity) are typically handled (and desired). An extensive review
of the existing isogeometric techniques for trimmed domains, mostly in
the context of trimmed surfaces, is presented in~\cite{Marussig2018}.

\NewT{Meshing tools and algorithms are a long and large subject of
research in the finite element community.  For example,
\cite{geuzaine2009gmsh} introduce {\em Gmsh}, a commonly used tool for
mesh generation that is suitable for finite element analysis,
including of \Bspline{} surfaces.  Starting from a B-rep model, the
boundary (trimmed) surfaces are triangulated, and then methods for
tetrahedral mesh generation are applied to generate finite elements
approximating the enclosed volume of the B-rep model. In addition to
the drawbacks discussed above when applying tetrahedral mesh
generation algorithms, the generated elements in Gmsh only approximate
the boundary surfaces, and an optimization process needs to be applied
in order to reduce the approximation error. Specifically, the tool is
not sufficient for IGA needs, as in IGA, the integration should be
performed in the parametric domain of the underlying trivariate, where
the integration boundaries shouldn't cross knot values. Hence, method
for subdividing the trimmed \Bspline{} trivariate into trimmed
\Bezier{} trivariates as well as further processing of trimmed
\Bezier{} trivariates are required, and such tools are not available
in Gmsh.}

Other studies, such as~\cite{MARTIN2010187}, suggests volumetric
spline parametrization of a (B-rep) polygonal mesh. Although
generating high order trivariates, and since the input is a discrete
triangulated mesh, it is not clear how to utilize such an approach in
a design process, while~\cite{MARTIN2010187} can serve as another way
of synthesizing trivariates.

In \cite{xu2013analysis}, a volumetric parametrization of a multi-block computational domain is provided, where the computation domain is segmented into several blocks, and each block is bounded by six tensor product 2D faces that are used to generate a trivariate. An iterative process of moving the trivariate's control points is then applied in order to achieve a better analysis suitable parametrization. The optimization process improves uniformity of the Jacobian for each trivariate as well as smoothness between trivariates. However, the method in \cite{xu2013analysis} assumes full computational domains, and doesn't handle trimmed trivariates.

\cite{BPlus_2D} propose B++ splines; trimmed B-spline basis functions that are developed to perform IGA on trimmed B-rep models. In order to handle arbitrary trimming curves, the trimmed spline basis functions 
are defined over samples of trimming curves. Such sampling introduces inaccuracies, and trimmed surfaces with sharp and complex boundaries require dense sampling in order to achieve a certain accuracy. The inaccuracy is delegated to the analysis stage as the method in~\cite{BPlus_2D} requires conversion of the B-rep trimmed spline surfaces into B++ surfaces.

In~\cite{liu2014volumetric_tspline_csg}, T-spline trivariates are extracted from boundary triangulated surfaces, using hierarchical octrees, and Boolean operation of volumetric cylinders and cubes. However, even simple models composed of  cones and tetrahedrons cannot be handled by~\cite{liu2014volumetric_tspline_csg}.

An analysis suitable trivariate parametrization for a B-rep model is suggested in \cite{ENGVALL201783}. First, the B-rep model is approximated as a polygonal mesh, then the interior of the polygonal mesh is covered with unstructured tetrahedrons, and finally each tetrahedron is replaced with a rational \Bezier{} trivariate with, possibly, faces extracted from the original trimmed surfaces. 
Unfortunately, the method loses the connection to the input B-rep surfaces, and is largely ignoring the problem of trimming, and assumes non-trimmed surfaces of the B-rep model and thus can not handle trimmed trivariates.

To summarize, previous methods of trivariate volumetric parametrization of general models lack precision or robustness, and cannot handle trimmed \Bspline{} volumes. Additionally, previous methods suggest a solution for B-rep models, and \NewT{none} of them suggest a solution for a general trimmed trivariate, that can result from volumetric Boolean operations.

\section{The Untrimming Algorithm}	\label{sec:algoSec}
In this section, we present the volumetric untrimming algorithm.
Recall Figure~\ref{fig:bez_subd_trimmed_trivar}(a) and consider a trimmed \Bspline{} trivariate ${\cal T}$, following Definition~\ref{def:trimmed_tv}, as a tensor product \Bspline{} trivariate $T$, and a set of trimming (trimmed) surfaces, ${\cal S}$. We assume the following on ${\cal T}$, as in~\cite{Massarwi2016}:
\begin{itemize}
\item ${\cal S}$ consists of (trimmed) surfaces that together form a 2-manifold B-rep model in the domain of $T$. Note that each surface $S\in{\cal S}$ is a trimming surface of $T$, but $S$ can also be a trimmed surface on its own. 
\item Each trimmed surface ${S \in \cal S}$ is orientable and the normal at each point on $S$ points toward the inside of ${\cal T}$.
\end{itemize}

The untrimming algorithm decomposes ${\cal T}$ into a set of mutually exclusive tensor product trivariates that covers the volume enclosed by ${\cal T}$ and precisely interpolates the boundary (trimming surfaces). The untrimming can be performed either in the Euclidean space, where the algorithm is applied directly over the 3D geometry of ${\cal T}$, or in the parametric space of $T$, where the algorithm is applied on the trimmed parametric domain of $T$. The result of the untrimming in the parametric space, is a set of mutually exclusive tensor product trivariates $\tau$, such that $T(\tau)$ completely covers ${\cal T}$, in the Euclidean space. 

Some simple applications, such as enclosed volume computation, can be performed by summing the enclosed volume of each of the tensor product trivariates of the untrimming result in the Euclidean space. However, in general integration applications, the untrimming needs to be performed in the trimmed parametric domain.
In the next sub-sections, we describe the untrimming algorithm in the parametric space (i.e the final output is a set of tensor product \Bspline{} trivariates in the parametric space of $T$, that covers the trimmed domain). The untrimming in the Euclidean space can be considered as a special case of the algorithm, having the bounding box of ${\cal T}$ as the associated trivariate $T$.

The algorithm consists of two main stages: In the first stage, discussed in Section~\ref{sec:bezierSubdiv}, ${\cal T}$ is subdivided into trimmed \Bezier{} trivariates at all the knot planes of $T$. In the second stage, discussed in Section~\ref{sec:bezierUntrim}, each trimmed \Bezier{} trivariate from the first stage is decomposed into a set of mutually exclusive (possibly singular on the boundary) tensor product \Bspline{} trivariates that cover the domain.

\subsection{Subdivision into \Bezier{} trimmed trivariates} \label{sec:bezierSubdiv}
The integration of \Bspline{} basis functions is needed in IGA, for example, in the computation of the mass and stiffness matrix elements~\cite{bartovn2017efficient}. In order to use the Gaussian quadrature method~\cite{atkinson2008introduction} for precise computation of the integral over the \Bspline{} basis functions, the \Bspline{} basis functions need to be split into polynomial or rational (in the \NURB{} case) functions. \Bezier{} elements extraction of a tensor product \Bspline{} functions can be performed by multiple knot insertion~\cite{cohengeometric}. However, for the \Bezier{} elements' extraction from a trimmed \Bspline{} trivariate ${\cal T}$, the trimming (trimmed) surfaces, ${\cal S}$, need to be subdivided  as well, along the iso-parametric surfaces of $T$, at all the internal knots of $T$.

\begin{algorithm}
	\textbf{Input}:\\
	\begin{tabularx}{\textwidth}{ll}
	${\cal T}=\{T,{\cal S}\}$: &A trimmed \Bspline{} trivariate; $T$ is a tensor product \Bspline{} trivariate, that is trimmed by \\ & set of trimmed surfaces ${\cal S}$;\\
	$t, dir$: & Subdivision parameter, $t$, of $T$, in parametric direction $dir$: $u,v$ or $w$; 
	\end{tabularx}
	\textbf{Output}:\\
	\begin{tabularx}{\textwidth}{ll}
	${\cal T}_t^{dir}$:& A set of trimmed trivariates which are the result of subdividing ${\cal T}$ at parameter $t$ in direction $dir$;
	\end{tabularx}
	\mbox{\vspace{-0.15in}}\\[-0.15in]
	// The following auxiliary function constructs a set of trimmed trivariates  
	   by trimming the trivariate $T$,\\
	// with a set of trimming (trimmed) surfaces, ${\cal S}$. \\  
	Function \textbf{GroupNewTrimTV(${\cal S}, T$)}
	\begin{algorithmic}[1]
		\STATE ${\cal T}_{Res} := \emptyset$;
		\IF {${\cal S} \neq \emptyset$}
		\STATE ${\cal {CS}}$ := Group ${\cal S}$ into sets of 2-manifold connected closed \NewT{component} surfaces;
		\FORALL{${\cal S}_i \in {\cal {CS}}$}
		\STATE ${\cal T}_i$ := Trimmed Trivariate \{$T$,${\cal S}_i$\};
		\STATE ${\cal T}_{Res} := {\cal T}_{Res} \bigcup~\{{\cal T}_i\}$;
		\ENDFOR 
		\ENDIF
		\RETURN ${\cal T}_{Res}$;
	\end{algorithmic}
	\mbox{\vspace{-0.35in}}\\[-0.35in]
	\textbf{\\ Algorithm}:
	\begin{algorithmic}[1]
		\STATE $S_{iso}$ := Iso surface of $T$ at parameter $t$, in direction $dir$;
		\STATE $\{T_L, T_R\}$ := Subdivided $T$ at parameter $t$, in direction $dir$;
		\STATE ${\cal S}_L$ := Intersect(${\cal S}$,$S_{iso}$); //B-rep Boolean ${\cal S}*S_{iso}$ \label{SubivSL}
		\STATE ${\cal S}_R$ := Subtract(${\cal S}$,$S_{iso}$);  //B-rep Boolean ${\cal S}-S_{iso}$\label{SubivSR}
		
		\STATE ${\cal T}_t^{dir} := \begin{aligned}
			&\textbf{GroupNewTrimTV}({\cal S}_L, T_L)~ \cup \\&\textbf{GroupNewTrimTV}({\cal S}_R, T_R);
		\end{aligned}$
		
		\RETURN ${\cal T}_t^{dir}$;
	\end{algorithmic}
	\caption{\bf TrimTrivarSubdiv(${\cal T}, t, dir$): Subdivision of a trimmed \Bspline{} trivariate ${\cal T}$ at $t$ along $dir$ }
	\label{alg:TrimTVSubdAlgo}
\end{algorithm}

\begin{algorithm}
	\textbf{Input}:\\
		\begin{tabularx}{\textwidth}{ll}
		${\cal T}=\{T,{\cal S}\}$: &A trimmed \Bspline{} trivariate; $T$ is a tensor product \Bspline{} trivariate, that is trimmed by \\ & set of trimmed surfaces ${\cal S}$;
	\end{tabularx}
	
	\textbf{Output}:\\
	\begin{tabularx}{\textwidth}{ll}
	${\cal B}$:& A set of trimmed \Bezier{} trivariates;
	\end{tabularx}
	
	\textbf{Algorithm}:
	\begin{algorithmic}[1]
		\IF [no internal knots]{$T$ is a \Bezier{} trivariate}
		\RETURN ${\cal T}$;
		\ELSE
		\STATE $t, dir$ := Find an internal knot, $t$, of $T$, in parametric direction $dir$ ($u,v$ or $w$); \mbox{\hspace{0.6in}}
		\STATE ${\cal T}_t^{dir}$ := \textbf{TrimTrivarSubdiv}(${\cal T},t,dir$);
		\STATE ${\cal B}$ := $\emptyset$;
		\FORALL{${\cal T}_i \in {\cal T}_t^{dir}$}
		\STATE ${\cal B}$ := ${\cal B} \bigcup$~ \textbf{TrimTVBezierSubd}(${\cal T}_i$);
		\ENDFOR
		\RETURN ${\cal B}$;
		\ENDIF
	\end{algorithmic}
	\caption{\bf TrimTVBezierSubd(${\cal T}$): Subdivision of a trimmed trivariate ${\cal T}$ into trimmed \Bezier{} trivariates }
	\label{alg:BezierSubdAlgo}
\end{algorithm}

Algorithm~\ref{alg:TrimTVSubdAlgo} presents the process of subdividing a given trimmed trivariate, at a given parametric value and direction. 
Algorithm~\ref{alg:BezierSubdAlgo} describes the subdivision algorithm of trimmed trivariate ${\cal T}$, into trimmed \Bezier{} trivariates.   Algorithm~\ref{alg:BezierSubdAlgo} recursively divides the trimmed trivariate ${\cal T}$, at all the internal knots of $T$, in all parametric directions, using Algorithm~\ref{alg:TrimTVSubdAlgo}. 

The general trimmed trivariate subdivision algorithm
(Algorithm~\ref{alg:TrimTVSubdAlgo}) is performed in the Euclidean
space, as the trimming surfaces ${\cal S}$ are given,
following~\cite{Massarwi2016}, in the Euclidean space. In
Algorithm~\ref{alg:TrimTVSubdAlgo}, $T$ is subdivided at a parametric
value $t$, in direction $dir$, resulting in two tensor product
trivariates $T_L$ and $T_R$. The trimming surfaces, ${\cal S}$, are
then separated by the iso-surface value $t$, in direction $dir$,
$S_{iso}$, into two groups. The separation is done by applying B-rep
intersection and subtraction Boolean
operations~\cite{bops_Satoh:1991:BOS, bops_thomas1986set} between
${\cal S}$ (that forms a closed B-rep), and $S_{iso}$. See
lines~\ref{SubivSL} and~\ref{SubivSR} in
Algorithm~\ref{alg:TrimTVSubdAlgo}, where ${\cal S}_L$ and ${\cal
S}_R$ are the result of the intersection and subtraction Boolean
operations between ${\cal S}$ and $S_{iso}$ respectively. \NewT{Such
Boolean set operations typically involve computing surface-surface
intersections (SSI) that may introduce approximation errors. However,
methods for computing SSI have been investigated for a long time
(i.e.~\cite{GRANDINE1997111}), including bounds on errors in the SSI.}
Each iso-parametric surface of $T$, divides $T$ into two
parts. However, since the trimming can impose arbitrary topology, each
of the two groups (${\cal S}_L$ and ${\cal S}_R$), from both sides of
$S_{iso}$, can consist of multiple connected components of trimming
surfaces, see for example Figure~\ref{fig:bez_subd_multiple_cc}. For
each group of trimming surfaces that forms a connected component,
${\cal S}_M$, a trimmed trivariate is constructed in
\textbf{GroupNewTrimTV}, having ${\cal S}_M$ as the trimming surfaces
and either $T_L$ or $T_R$ (the one that contains ${\cal S}_M$) as its
tensor product trivariate. See
Figures~\ref{fig:bez_subd_trimmed_trivar}-\ref{fig:bez_subd_sphere}
for some results of Algorithm~\ref{alg:BezierSubdAlgo}.

\begin{figure}
\centering
  \begin{minipage}{.48\textwidth}
      \centering
      \begin{tabular}{c}
          \includegraphics[scale = 0.25]{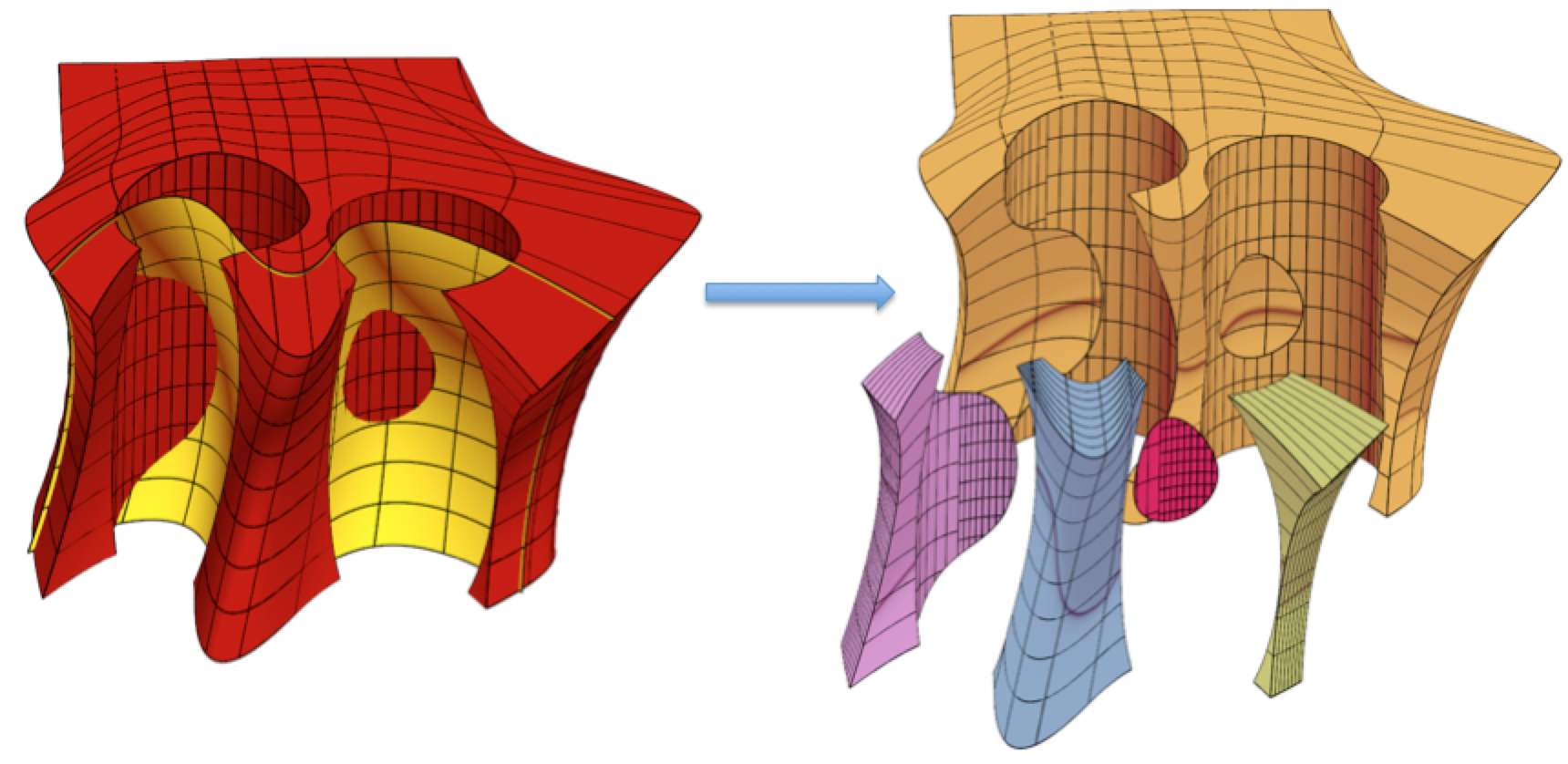}
      \end{tabular}
      \begin{picture}(0,0)
	\put(-65, 15){$(a)$}
	\put( 10, 95){$(b)$}
	\put( 50,  5){$(c)$}
      \end{picture}
      \mbox{\vspace{-0.14in}}\\[-0.14in]
      \caption{Subdivision of trimmed trivariate (red) in (a) along an iso-parametric surface (yellow) results in five trimmed trivariates: one trimmed trivariate on the back side of the iso-surface in (b), and four trimmed trivariates on the front side of the iso-surface in (c). An exploded view is presented in (b)/(c).}
     \label{fig:bez_subd_multiple_cc}
  \end{minipage}
  \mbox{\hspace{0.1in}}
  \begin{minipage}{.48\textwidth}
    \centering
	\begin{tabular}{cc}
		\mbox{\hspace{-0.21in}}
		\includegraphics[scale = 0.16]{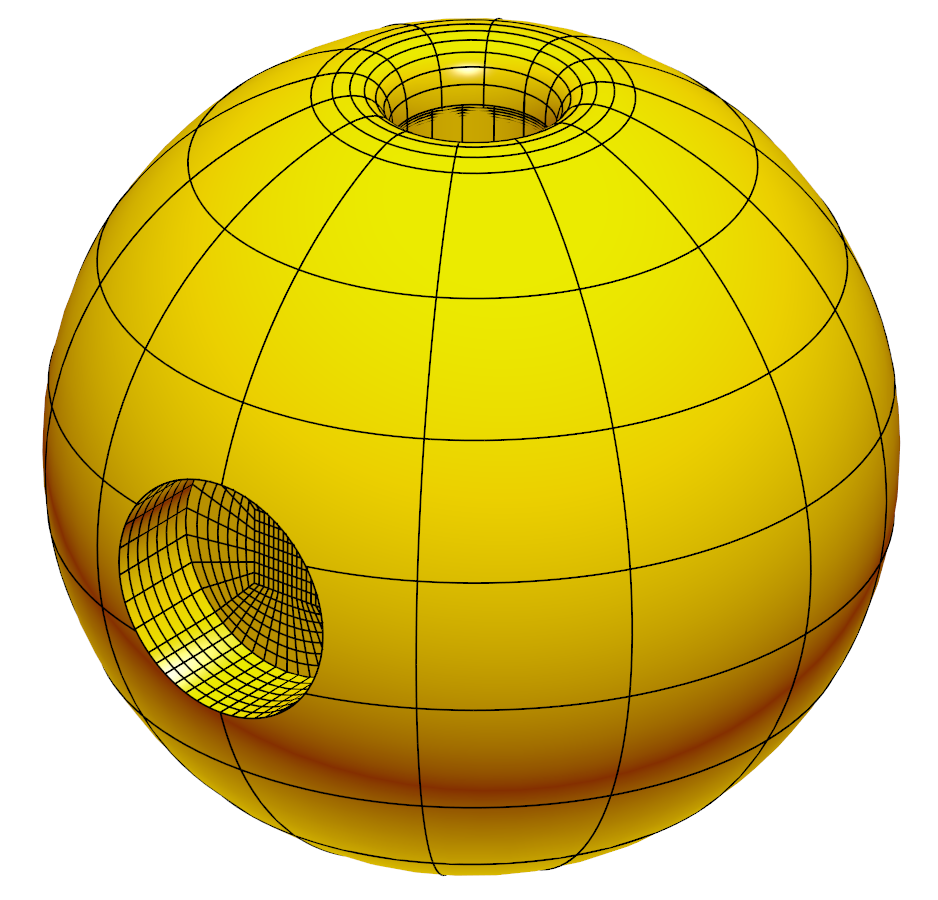} &
		\mbox{\hspace{-0.17in}}
		\includegraphics[scale = 0.14]{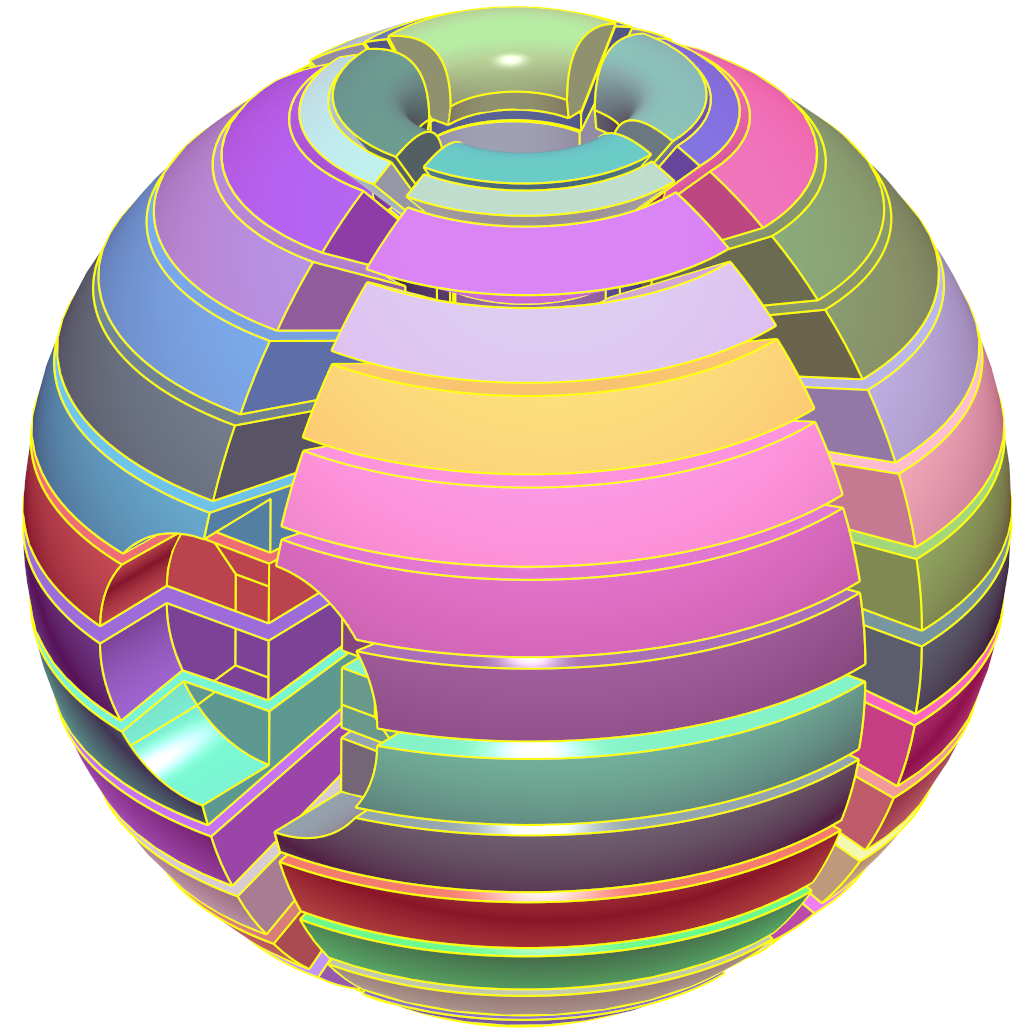}
	\end{tabular}
	\begin{picture}(0,0)
	  \put(-20, 15){$(a)$}
	  \put(100, 15){$(b)$}
	\end{picture}
	\mbox{\vspace{-0.14in}}\\[-0.14in]
	\caption{Subdivision of trimmed \Bspline{} trivariates into trimmed \Bezier{} trivariates. (a) Ten trimmed \Bspline{} trivariates forming a spherical volumetric model. (b) Subdivision of the trimmed trivariates in (a) into 72 trimmed \Bezier{} trivariates (an exploded view).}
     \label{fig:bez_subd_sphere}
  \end{minipage}
\end{figure}

\subsection{Untrimming a \Bezier{} trimmed trivariate}  \label{sec:bezierUntrim}

In this section, we describe an algorithm to decompose a trimmed \Bezier{} trivariate ${\cal T}$ into a set of possibly singular on the boundary tensor product \Bspline{} trivariates that covers ${\cal T}$. Recall ${\cal T}$ is represented as a tensor product \Bezier{} trivariate $T$, and a set of trimming (trimmed) \Bspline{} surfaces, ${\cal S}$. The trimmed surfaces in ${\cal S}$ can be classified into two groups: The first group contains boundary surfaces of $T$ that are possibly trimmed, and the second group contains surfaces that are general trimming (trimmed) surfaces in the Euclidean space that can be, for example, a result of Boolean operation with other (trimmed) trivariate. Since we aim for IGA applications that require integration over the trimmed trivariate's domain in the parametric space, the Euclidean space surfaces in ${\cal S}$ need to be back projected into the parametric domain of $T$. However, trimming surfaces in the second group can not, in general, have an algebraic form and can not have a precise piecewise-polynomial image in the parametric space of $T$. In such cases, the trimming surfaces must be approximated in the parametric space, for example, by a least squares fit.
%\newline

Algorithm~\ref{alg:BezierUntrimAlgo} describes the untrimming algorithm of a trimmed \Bezier{} trivariate; decomposing a trimmed \Bezier{} trivariate into a set of mutually exclusive (singular only on the boundary) tensor product \Bspline{} trivariates $\tau$ that cover the trimmed parametric domain of ${\cal T}$, where $T(\tau)$ covers ${\cal T}$ in the Euclidean space.  
The key idea of the algorithm is finding a kernel point of ${\cal T}$: an internal point in ${\cal T}$ that is visible from all
points on all trimming surfaces in ${\cal S}$. If such a kernel point, $P$, exists, the set of trimming (trimmed) surfaces of ${\cal T}$, ${\cal S}$, is untrimmed: ${\cal S}$ is decomposed into a set of tensor product \Bspline{} surfaces, $\overline{{\cal S}}$, by surface untrimming operations, for example following~\cite{massarwi2018untrimming}. A singular (only on its boundary) tensor product \Bspline{} trivariate is then constructed between $P$ and each surface in $\overline{{\cal S}}$. 
Alternatively, If no kernel point exists, ${\cal T}$ is subdivided along an iso parametric direction, and the algorithm is recursively applied on each part. 

The algorithm suggests a solution for the 3D art-gallery~\cite{marzal2012three,berg2008computational} for objects having free form (trimmed) boundary surfaces, where each kernel point can be treated as a guard. However, clearly, the algorithm is not optimal in terms of minimal number of kernel points. Finding the optimal solution of the art gallery problem, even for planar polygons, is NP-hard~\cite{lee1986computational}

The algorithm is iterative; each iteration consists of four steps.~In the first step, discussed in Section~\ref{sec:Untrimming_surfaces}, ${\cal S}$ is untrimmed to tensor product \Bspline{} surfaces, $\overline{{\cal S}}$ (see line~\ref{bezUntrimUntrimming} in Algorithm~\ref{alg:BezierUntrimAlgo}). In the second step, discussed in Section~\ref{sec:ParamSpaceApprox}, the untrimmed tensor product surfaces from the previous step are approximated in the parametric space of $T$
(see line~\ref{bezUntrimApprox} in Algorithm~\ref{alg:BezierUntrimAlgo}). In the third step, discussed in Section~\ref{sec:Kernel_Point}, the algorithm seeks a  kernel point in the domain of ${\cal T}$, and if such a kernel point is found, builds mutually exclusive tensor product trivariates that cover ${\cal T}$ (see lines~\ref{Alg:FindKernelStart}-\ref{Alg:FindKernelEnd} in Algorithm~\ref{alg:BezierUntrimAlgo}). In the fourth step, discussed in Section~\ref{sec:No_Kernel_Subd}, if no kernel point is found, ${\cal T}$ is subdivided into several parts (using Algorithm~\ref{alg:TrimTVSubdAlgo}), and the untrimming algorithm is recursively invoked on each part (see lines~\ref{Alg:NoKernelStart}-\ref{Alg:NoKernelEnd} in Algorithm~\ref{alg:BezierUntrimAlgo}). Section~\ref{sec:No_Kernel_Subd} also discusses the termination of this algorithm.

\begin{algorithm}
	\textbf{Input}:\\
	\begin{tabularx}{\textwidth}{ll}
		${\cal T}=\{T,{\cal S}\}$: & a trimmed \Bezier{} trivariate; $T$ is the tensor product \Bezier{} trivariate, that is trimmed by\\
		& set of trimmed surfaces ${\cal S}$;
	\end{tabularx}
	\textbf{Output}:\\
	\begin{tabularx}{\textwidth}{ll}
	$\tau$: & a set of mutually exclusive (singular only on the
	          boundary) tensor product \Bspline{} trivariates that\\
	          &covers the parametric domain of ${\cal T}$;
	\end{tabularx}
	\textbf{Algorithm}:
	\begin{algorithmic}[1]
	\STATE $\tau$ := $\emptyset$;
	\STATE $\overline{{\cal S}}$ := Untrimming trimmed surfaces in ${\cal S}$ into tensor product surfaces; // i.e.~\cite{massarwi2018untrimming}\label{bezUntrimUntrimming}
	\STATE $\overline{{\cal S}}^T$ := Approximation of $T^{-1}(\overline{{\cal S}})$; \label{bezUntrimApprox}

		\STATE $p$ := Find a kernel point of $\overline{{\cal S}}^T$; \label{Alg:FindKernelStart}
		\IF {$p \neq null$}
		\FORALL{$s_i(u,v) \in \overline{{\cal S}}^T$}
			\STATE $U_i(u,v,w) := (1-w)p + ws_i(u,v),~w\in[0,1]$;
			\STATE $\tau := \tau \bigcup \{U_i(u,v,w)\}$; 
		\ENDFOR \label{Alg:FindKernelEnd}
		\ELSE 
			\STATE $t, dir$ := Find a subdivision parametric value $t$, for $T$, in direction $dir$;\label{Alg:NoKernelStart}
			\STATE ${\cal T}_t^{dir}$ := \textbf{TrimTrivarSubdiv}(${\cal T}$, $t$, $dir$); //Algorithm~\ref{alg:TrimTVSubdAlgo}
			
			\FORALL{${\cal T}_i \in {\cal T}_t^{dir}$}
			\STATE $\tau := \tau \bigcup$~\textbf{TrivBezierUntrim}(${\cal T}_i$);
			\ENDFOR \label{Alg:NoKernelEnd}
		\ENDIF
		\RETURN $\tau$;  
	\end{algorithmic}
	\caption{\bf TrivBezierUntrim(${\cal T}$): Decomposing trimmed \Bezier{} trivariate ${\cal T}$ into (singular only on the boundary) tensor product \Bspline{} trivariates.}
	\label{alg:BezierUntrimAlgo}
\end{algorithm}

\subsubsection{Untrimming of the trimming surfaces in ${\cal S}$} \label{sec:Untrimming_surfaces}
At each iteration of Algorithm~\ref{alg:BezierUntrimAlgo}, the trimming surfaces of ${\cal T}$, ${\cal S}$, which can be trimmed surfaces, are untrimmed in Euclidean space. In the process of untrimming trimmed surface $S_i \in {\cal S}$, $S_i$ is precisely decomposed  into a set of tensor product \Bspline{} surfaces that covers $S_i$. \NewT{In this work,} we use the minimal weight untrimming algorithm proposed in~\cite{massarwi2018untrimming}; \NewT{though, any other untrimming algorithm of trimmed surfaces can be employed as well. The untrimming algorithm in ~\cite{massarwi2018untrimming}} guarantees positive determinant of the Jacobian in the interior of the resulting tensor product surfaces, provided the input is regular. The untrimming of ${\cal S}$ is needed for the next steps since it simplifies the visibility test used in seeking a kernel point (see Section~\ref{sec:Kernel_Point}), as visibility test for a tensor product \Bspline{} surface is simpler than the same test for a trimmed surface - where an intersection between trimming curves and the boundaries of the visible regions of the underlying tensor product surface should be computed.
%Second, having tensor product surfaces is required for the approximation of the trimmed surface in the parametric space of $T$, as in general, the original tensor product surfaces in ${\cal S}$ may not be contained entirely in the domain $T$. Third, the decomposition of ${\cal S}$ into tensor product \Bspline{} surfaces is required for constructing tensor product \Bspline{} trivariates, once a kernel point has been computed (see Section~\ref{sec:Kernel_Point}).

\subsubsection{Approximating ${\cal S}$ in the parametric space of $T$}
\label{sec:ParamSpaceApprox}

After untrimming ${\cal S}$ into tensor product surfaces,
$\overline{{\cal S}}$, as described in
Section~\ref{sec:Untrimming_surfaces}, each tensor product surface
$\overline{S_i}\in\overline{{\cal S}}$ is least squares approximated
with a tensor product surface $s_i$ in the parametric space of $T$,
such that $\overline{S_i}\cong T(s_i)$. 
The approximation parameters of $s_i$ are (in
each parameter direction): The number of samples $m$ of
$\overline{S_i}$, the order $o$ of $s_i$, and the number of control
points, $n$, of $s_i$. The approximation algorithm consists of three
steps:
\begin{enumerate}
	\item Sample $m^2$ points on $\overline{S_i}$, $\{P_{jk}\}$.
	\item For each sample $P_{jk} = P_{jk}(x,y,z)$, find the back projected point $p_{jk}(u,v,w)$ in the parametric space of $T$, such that $T(p_{jk})=P_{jk}$. This can be done by finding the solution of the following system of three equations and three unknowns $(u,v,w)$:
\begin{equation}
\begin{aligned}
	&T_x(u,v,w) = x,~~~~~
	&T_y(u,v,w) = y,~~~~~
	&T_z(u,v,w) = z,
\end{aligned}
\label{eqn-pt-back-proj}
\end{equation}
where $T_x, T_y, T_z$ refer to the $x,y,z$ components of $T(u,v,w)$ respectively. Note that if $T$ is regular, only a single solution exists.
	\item $s_i \leftarrow$ Least squares approximation~\cite{atkinson2008introduction} of $\{p_{jk}\}$, by a tensor product \Bspline{} surface of size $(n\times n)$ and orders $(o \times o)$. 

\end{enumerate} 

\NewT{Quite a few methods exist for computing the precise (up to
machine-precision) back projections of points in freeforms. This
classic inverse problem is, for example, efficiently solved billions
of times in~\cite{Ezair2017FabricatingFG} for a regular trivariate
$T(u, v, w)$ by posing it as (and solving) three algebraic constraints
in $(u, v, w)$ (as in Equation~(\ref{eqn-pt-back-proj})).
Going from back projection of points to back projection of curves and
surfaces is more involved but relates to fitting freeforms to point
sets, methods that are well digested and are beyond the scope of this
work.}

The approximation process described above is applied only to surfaces that are not lying on the boundaries of $T$. In the case of a trimming surface that is a (trimmed) boundary surface of $T$, $s_i$ is simply extracted as that boundary. See an example in Figure~\ref{fig:approx_param_space}. Finally, note $T(s_i)$ can be precisely computed by a symbolic surface-trivariate composition~\cite{DeRose93,Elber92}.

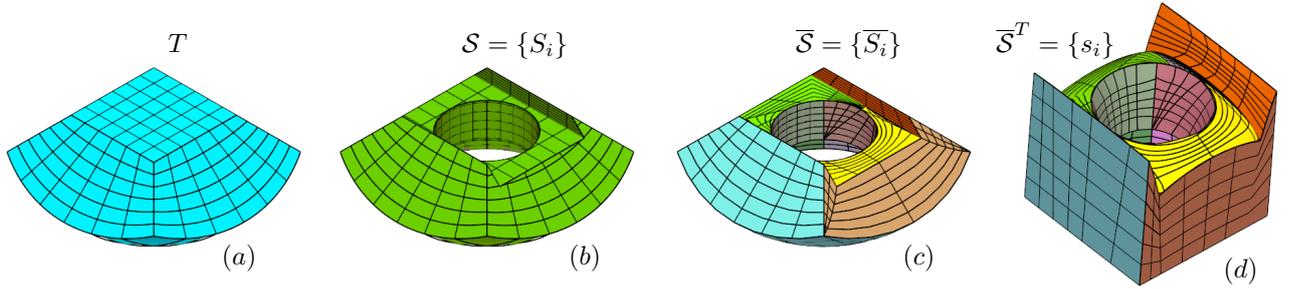
\begin{figure*}
	\begin{center}
		\begin{tabular}{cccc}
			\mbox{\hspace{-0.21in}}
			\includegraphics[scale = 0.34]{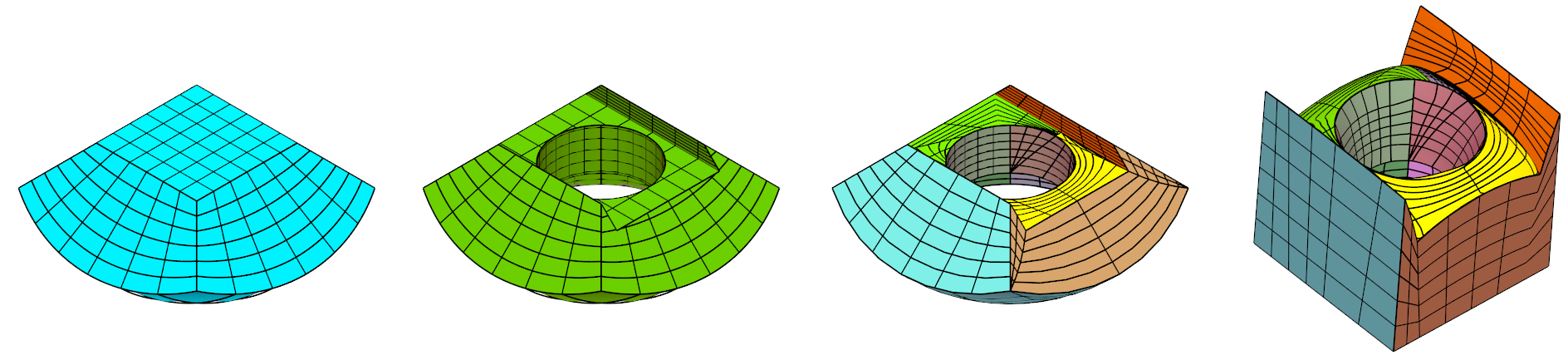} &
		\end{tabular}
	\end{center}
	\begin{picture}(0,0)
	\put( 80, 25){$(a)$}
	\put( 60, 105){$T$}
	
	\put(210, 25){$(b)$}
	\put(170, 105){${\cal S}=\{S_i\}$}
		
	\put(335, 25){$(c)$}
	\put(295, 105){$\overline{{\cal S}}=\{\overline{S_i}\}$}
		
	\put(455, 20){$(d)$}
	\put(370, 105){$\overline{{\cal S}}^T=\{s_i\}$}
	\end{picture}
	\mbox{\vspace{-0.5in}}\\[-0.5in]
	\caption{Untrimming and approximating trimming surfaces, ${\cal S}$, in the parametric space of trivariate $T$. (a) The trivariate $T$ in Euclidean space. (b) The trimming surfaces, ${\cal S}$, also in Euclidean space. (c) $\overline{\cal S}$, the untrimming of ${\cal S}$ in the Euclidean space. (d) Approximation of $\overline{\cal S}$ in the parametric space of $T$, by tensor product surfaces of orders 3x3.}
	\label{fig:approx_param_space}
\end{figure*}

\subsubsection{Seeking a kernel point} \label{sec:Kernel_Point}
A kernel point of a closed B-rep model is an internal point that is visible from every point on the boundary of the model~\cite{berg2008computational}. Consider the trimming surfaces of a trimmed trivariate, given as a set of  tensor-product \Bspline{} surfaces, $\overline{{\cal S}}^T$, (following the untrimming process in Section~\ref{sec:Untrimming_surfaces}, and the back projection to the parametric space in Section~\ref{sec:ParamSpaceApprox}). 
\begin{lemma} \label{lemma:vis}
Assume every surface $s_i\in\overline{{\cal S}}^T$ is $C^1$. A point $p$ is a kernel point of the B-rep bounded by $\overline{{\cal S}}^T$, if the following condition is satisfied:
\begin{equation} \label{eqn:visibility}
	\langle p-s_i(u,v),n_i(u,v)) \rangle > 0, \forall s_i\in \overline{{\cal S}}^T, \forall u,v \in s_i,
\end{equation} 
where $n_i(u,v)=\frac{\partial s_i}{\partial u} \times \frac{\partial s_i}{\partial v}$ is the (unnormalized) normal field of $s_i$, pointing into the B-rep. Again, note $s_i$ is regular in its interior and hence $n_i$ never vanishes in the interior of the domain.
\end{lemma}
\begin{proof}
As stated at the beginning of Section~\ref{sec:algoSec}, we assume $n_i(u,v)$ points inside ${\cal T}$. Under this assumption, an internal point $p$ is visible to a point on a surface $s_i(u,v)$, if the straight line segment between $p$ and $s_i(u,v)$ is in the positive half plane defined by the normal $n_i(u,v)$ and the point $s_i(u,v)$, meaning that the vector $p-s_i(u,v)$ and the normal $n_i(u,v)$ satisfy: $\langle p-s_i(u,v),n_i(u,v) \rangle > 0$.
Moreover, the line segment $\overline{p, s_i(u,v)}$ lies entirely inside ${\cal T}$.

Now, given an internal point $p$ that satisfies Equation~(\ref{eqn:visibility}), and assume, by contradiction, that there exists a point $s_j(u,v)$ that is not visible to $p$. Examine the line segment $L=\overline{p,s_j(u,v)}$. Since $s_j(u,v)$ is not visible to $p$, then, by the Jordan-Brouwer separation theorem~\cite{lima1988jordan}, and ignoring tangential contact(s) for now, there exists a finite segment of $L$ that lies entirely outside ${\cal T}$. Otherwise, if $L$ is entirely inside ${\cal T}$, then $s_j(u,v)$ is visible to $p$.
Let $\overline{s_l,s_k}$ be the segment of $L$ that is outside ${\cal T}$, where $s_l$ is the closest point to $p$ (See illustration on Figure~\ref{fig:visibility_crv}). 
Because, $\overline{s_l,s_k}$ stabs the boundary at $s_k$ from {\em outside}, $\langle p-s_k,n_k \rangle$ can not be positive, contradicting the assumption that $p$ satisfies Equation~(\ref{eqn:visibility}), for all $s_i$. With special care, and following similar lines, one can also handle tangential contacts.
\end{proof}

\begin{figure}
\centering
    \begin{minipage}{.48\textwidth}
        \centering
	\begin{tabular}{c}
		\mbox{\hspace{-0.21in}}
		\includegraphics[scale = 0.22]{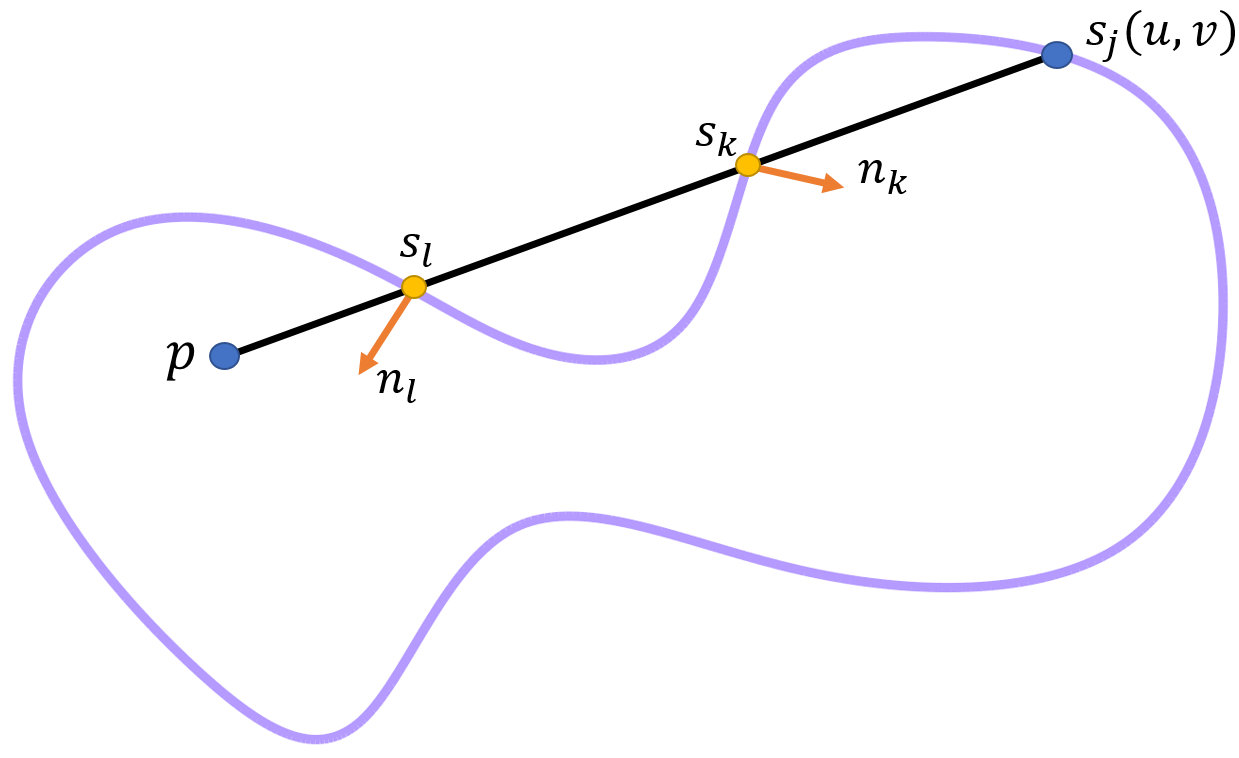}
	\end{tabular}
	\mbox{\vspace{-0.15in}}\\[-0.15in]
	\caption{Auxiliary illustration for Lemma~\ref{lemma:vis}.}
	\label{fig:visibility_crv}
    \end{minipage}
    \mbox{\hspace{0.1in}}
    \begin{minipage}{.48\textwidth}
        \centering
	\begin{tabular}{c}
		\mbox{\hspace{-0.21in}}
		\includegraphics[scale = 0.2]{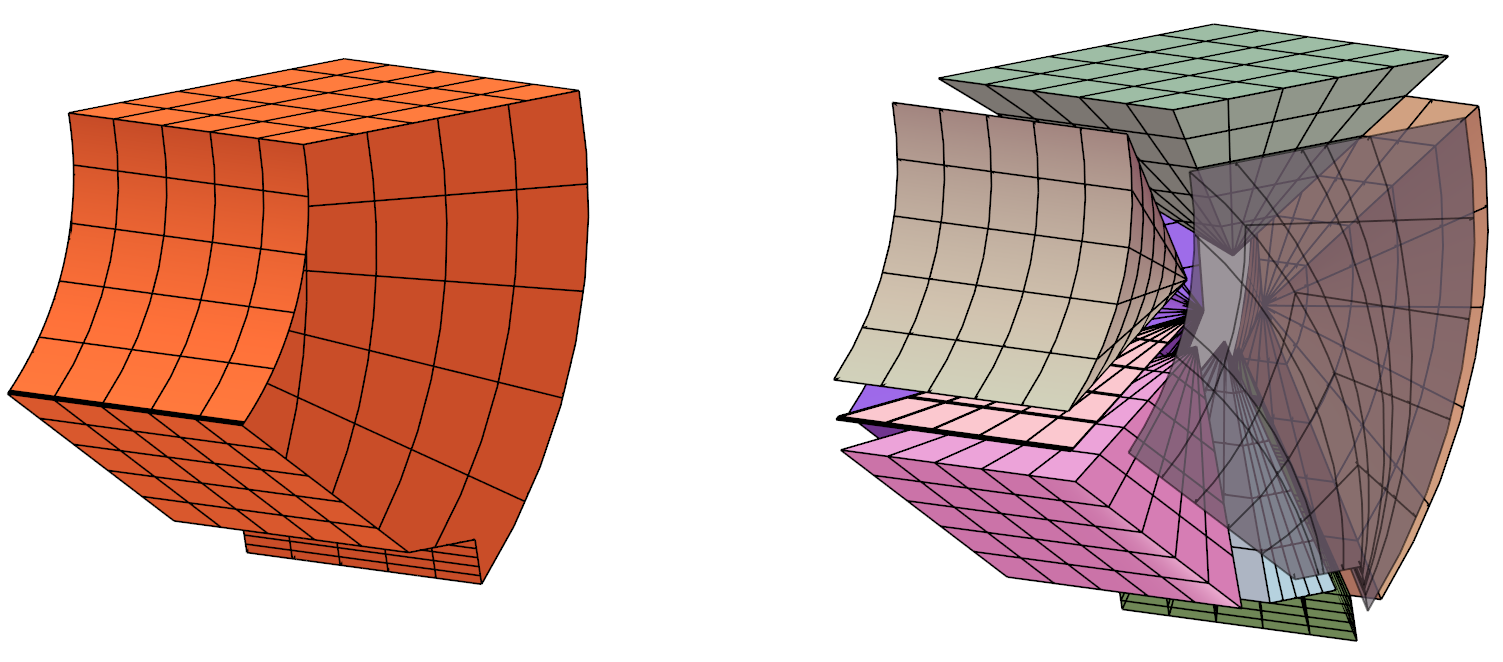} 
	\end{tabular}
	\begin{picture}(0,0)
	    \put(-35, 10){$(a)$}
	    \put(100, 10){$(b)$}
	\end{picture}
	\mbox{\vspace{-0.15in}}\\[-0.15in]
	\caption{Untrimming of a trimmed \Bezier{} trivariate, after finding a kernel point. A trimmed rational \Bezier{} trivariate of orders (4,2,2) having nine trimming trimmed surfaces in (a). Once a kernel point is found, nine (singular on the boundary) tensor product trivariates are constructed (b), and shown in an exploded view.}
	\label{fig:bez_kernel_untrim}
    \end{minipage}
\end{figure}

Lemma~\ref{lemma:vis} states that if all surfaces $s_i\in\overline{{\cal S}}^T$ are front facing with respect to $p$, then all $s_i\in\overline{{\cal S}}^T$ are visible to $p$, in a similar way to front face visibility in projections, in graphics.

Note that although the normal of the entire closed B-rep bounded by $\overline{{\cal S}}^T$ is not defined on intersection/stitching $C^0$ points between surfaces in $\overline{{\cal S}}^T$, the proof is still valid on such points, as Equation~(\ref{eqn:visibility}) still holds for each $s_i$  surface.

More details on freeform surface visibility can be found  in~\cite{elber1995arbitrarily}.
Equation~(\ref{eqn:visibility}) can be computed by symbolically representing this inner product ~\cite{Elber92} as a tensor product \Bspline{} function, and verifying the positivity of all the control coefficients.

As a side comment, the kernel, if any, can be bound using an approach that utilizes the intersections of all tangent planes at all parabolic points of $\overline{{\cal S}}$, assumed $C^2$~\cite{cohengeometric}. However, such an approach, that involves computing parabolic points, is computationally expensive and numerically challenging. Instead, and since herein we only seek a single kernel location, the trimmed parametric domain of ${\cal T}$ is \NewT{uniformly} sampled, and the visibility test of Equation~(\ref{eqn:visibility}) is verified on each sample. There could be more than one kernel sample that satisfies Equation~(\ref{eqn:visibility}). Among these samples, we, heuristically, pick the kernel point as the sample, $p$, that minimizes the following expression:
\begin{equation} \label{eqn:kernel_minmax_dist}
\max_{\forall s_i \in \overline{{\cal S}}^T}{(\min_{u,v}(\lVert s_i(u,v)-p\rVert))} - \min_{\forall s_j \in \overline{{\cal S}}^T}{(\min_{u,v}(\lVert s_j(u,v)-p\rVert))},
\end{equation}
that promotes kernel points for which the deviations of the maximal and minimal distance to the boundaries are minimal, as we try to avoid thin and tall trivariates in the output, as much as possible.

During the process of seeking a kernel point, as only $p$ changes during the domain sampling, caching computation results that are independent of $p$, such as $\langle s_i(u,v),n_i(u,v) \rangle$ in Equation~(\ref{eqn:visibility}), will result in a significant speed up improvement.

If a kernel point $p$ is found, a tensor product \Bspline{} trivariate $U_i(u,v,w)$ is constructed as a singular ruled trivariate~\cite{cohengeometric} between $p$ and each tensor product \Bspline{} surface $s_i(u,v) \in \overline{{\cal S}}^T$, as following:
\begin{equation}
	U_i(u,v,w) = (1-w)p + ws_i(u,v), w\in[0,1].
\end{equation} 

Since $s_i(u,v)$ is guaranteed to (possibly) have singularities only on the boundaries (property of the surfaces' untrimming algorithm~\cite{massarwi2018untrimming}), so is $U_i(u,v,w)$. See an example of untrimming of a trimmed \Bezier{} trivariate in Figure~\ref{fig:bez_kernel_untrim}.

Note that although the trimming curves are typically the result of Boolean set operations, and hence piecewise linear, the algorithm can handle smooth trimming curves of arbitrary order. The trivariates' construction procedures can handle higher order curves as is. However, a \Bspline{} curve fitting process \NewT{needs} to be applied in the back projection approximation in the parametric space, as discussed in Section~\ref{sec:ParamSpaceApprox}.

\subsubsection{Subdivision of ${\cal T}$} \label {sec:No_Kernel_Subd}
At each iteration, if no kernel point is found, ${\cal T}$ is subdivided 
into several trimmed trivariates, and the untrimming algorithm is recursively applied on each part. Several strategies can be utilized to choose the subdivision direction and value. In this work, the subdivision is performed, using Algorithm~\ref{alg:TrimTVSubdAlgo}, along the center of the bounding box of ${\cal T}$, in the parametric space, and the parametric direction is the direction of the longest dimension of the bounding box. This strategy ensure that the trimmed domain is getting smaller with each iteration.  

\NewT{During the iterative subdivision process, a non-trimmed (full tensor product) trivariate domains can be obtained, having six tensor product isoparametric boundary surfaces. In such cases, there is obviously no need to compute a kernel point and the trivariate is not subdivided further.}
Further, during the subdivision process, \NewT{not all boundary surfaces} are subdivided, in practice, in each iteration. Hence, we cache the untrimming and approximation results of the surfaces (results of steps~\ref{bezUntrimUntrimming} and~\ref{bezUntrimApprox} in Algorithm~\ref{alg:BezierUntrimAlgo},  respectively) to be used in \NewT{forthcoming} iterations.

We consider a special case that is handled differently: If the size of $s_i$, and the aperture of the normal cone of $s_i$ for each trimming surface, $s_i \in {\overline{\cal S}}^T$, is less than a predefined thresholds $\epsilon_e$ and $\epsilon_\theta$ respectively, then each $s_i$ is approximated by a planar surface, resulting in a polyhedron, $C$, that approximates ${\overline{\cal S}}^T$. Then, kernel points of $C$ can always be found, following ~\cite{marzal2012three} for example.
\begin{lemma}
	Assume all surfaces $s_i \in {\overline{\cal S}}^T$ have a bounded
        curvature.  If the subdivision process satisfies the following:
	\begin{itemize}
		\item The subdivision is performed at the center of the bounding box of ${\cal T}$ along the direction of the longest dimension of the bounding box, and, 
		\item ${\overline{\cal S}}^T$ is approximated by a polyhedron when the size of each $s_i \in {\overline{\cal S}}^T$ and the aperture of the normal cone of each $s_i \in {\overline{\cal S}}^T$, is less than a predefined thresholds, $\epsilon_e$ and $\epsilon_\theta$, respectively,
	\end{itemize}
	then,  Algorithm~\ref{alg:BezierUntrimAlgo} terminates after a finite number of iterations.
\end{lemma}
\begin{proof}
	Each surface $s_i \in {\overline{\cal S}}^T$ is represented as a \Bspline{} surface, and $s_i$ has a bounded curvature, by assumption. In other words, for each threshold $\epsilon_\theta$, there exists a threshold $\epsilon_e$ such that if the size of $s_i$ is less than $\epsilon_e$ (which can be achieved by a finite number of subdivisions), then the aperture of the normal cone of $s_i$ is less than $\epsilon_\theta$. Since the subdivision is performed along the longest dimension, it is guaranteed that the size of $s_i$ is getting smaller at each iteration. 

	Eventually, after a finite number of iterations, and if no kernel point is found, the size of each surface $s_i$ will be less than $\epsilon_e$ and the aperture of the normal cone of $s_i$ will be less than $\epsilon_\theta$. At that iteration, ${\overline{\cal S}}$ is approximated by a polyhedron and the algorithm terminates.
\end{proof}

\section{Results}	\label{sec:resultsSec}
The algorithms presented in this paper are all implemented in the IRIT~\cite{IRIT} solid modeling kernel. We now present results of the untrimming algorithm applied on several trimmed trivariates and V-rep models. The following examples were synthesized on a 3.4 GHz Intel i7 CPU with 32 GB RAM in a single thread mode and Windows 7. Practically, in all the examples presented in this paper, we never failed to find a kernel point and never had to resort to the polyhedron approximation (Recall Section~\ref{sec:No_Kernel_Subd}).

%\begin{table*}
%	\centering
%	\mbox{\hspace{-0.12in}}
%	\begin{tabular}{lcccc}
%		% centered columns
%		\hline\hline
%		%inserts double horizontal lines
%		Model                                  & \# Trivariates & \# Total domain       & Maximal         & Total time  \\ [0.5ex]
%		&            & subdivisions         & subdivision depth       & (Secs.)   \\ [0.5ex]
%		% inserts table
%		%heading
%		\hline
%		% inserts single horizontal line
%		Trimmed \Bezier{} trivariate  (Figure~\ref{fig:trivar2_untrim}) & TBD & TBD &     & 262          & 30  & 9 & 251      \\ [0.5ex]
%		V-rep model I   (Figure~\ref{fig:solid3_untrim})  & TBD & TBD & 76          & 5  & 3 & 52     \\ [0.5ex]
%		Trimmed \Bspline{} trivariate   (Figure~\ref{fig:trimmed5_untrim})      & 406         & 13  & 5   & 10.5     \\ [0.5ex]
%		V-rep model II  (Figure~\ref{fig:vsolid8_untrim})   & TBD & TBD  & 294          & 32  & 7   & 13.2    \\ [0.5ex]
%		V-rep model III (Figure~\ref{fig:vsolid9_untrim})   & TBD & TBD  & 370         & 22  & 4   & 20.9     \\ [0.5ex]
%		Trimmed unit cube (Figure~\ref{fig:volume_untrim})  & TBD & TBD  & 56         & 3  & 2   & 2.9     \\ [0.5ex]
%		\hline
%		%inserts single line
%	\end{tabular}
%	\caption{Statistics on the process of untrimming the trimmed trivariates 
%		presented in Figures~\ref{fig:trivar2_untrim}-\ref{fig:volume_untrim}.}
%	\label{table:vuntrimming_stats_old} 
%	% is used to refer this table in the text
%\end{table*}

\begin{table*}
	\small
	\centering
	\begin{tabular}{l|cc|ccccc}
		\toprule
		\multirow{3}{*}{Model} &
		\multicolumn{2}{c|}{Input} &
		\multicolumn{5}{c}{Output} \\\cline{2-8}
		& {\# Trimmed}& {\# Trimming} &  \# Tensor       & \# Total     & Maximal      & {Total } & (E)uclidean  \\
		& trivariates & {surfaces} &     product                   &  domain      & subdivision  &   time   &    or   \\
		&             &            &    Trivariates                    &  subdivisions  			     & depth        &  (Sec.)  & (P)arametric    \\
		\midrule
		Figure~\ref{fig:trivar2_untrim}    &    1   &    11    &   262   &   30   &   9   &    251    & P  \\ [0.5ex]
		Figure~\ref{fig:solid3_untrim}     &    1   &    24    &   76    &    5   &   3   &     52    & P  \\ [0.5ex]
		Figure~\ref{fig:trimmed5_untrim}   &    1   &    24    &   406   &   13   &   5   &   10.5    & E  \\ [0.5ex]
		Figure~\ref{fig:vsolid8_untrim}    &    7   &    64    &   294   &   32   &   7   &   13.2    & E  \\ [0.5ex]
		Figure~\ref{fig:vsolid9_untrim}    &   20   &   168    &   370   &   22   &   4   &   20.9    & E  \\ [0.5ex]
		Figure~\ref{fig:volume_untrim}     &    1   &    34    &    56   &    3   &   2   &    2.9    & E  \\ [0.5ex]
		\bottomrule
	\end{tabular}
	\caption{Statistics on the process of untrimming the trimmed trivariates 
		presented in Figures~\ref{fig:trivar2_untrim}-\ref{fig:volume_untrim}.}
	\label{table:vuntrimming_stats}
\end{table*}

Figure~\ref{fig:trivar2_untrim} shows the result of untrimming a trimmed \Bezier{} trivariate of orders (2,2,3). The untrimming is performed in the parametric space, and the trimming surfaces are untrimmed and approximated by a set of bi-quadratic surfaces. The result consists of 262 tensor product \Bspline{} trivariates that covers the trimmed parametric domain of the trimmed trivariate. The whole untrimming process took 251 seconds, out of which the approximation of the untrimmed surfaces in the parametric space took 107 seconds.

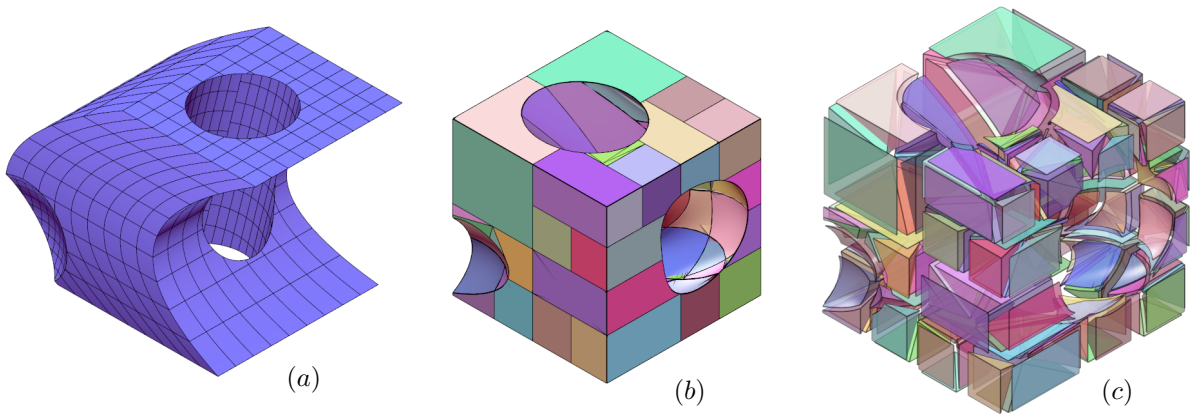
\begin{figure}
	\begin{center}
		\begin{tabular}{cccc}
			\mbox{\hspace{-0.21in}}
			\includegraphics[scale = 0.32]{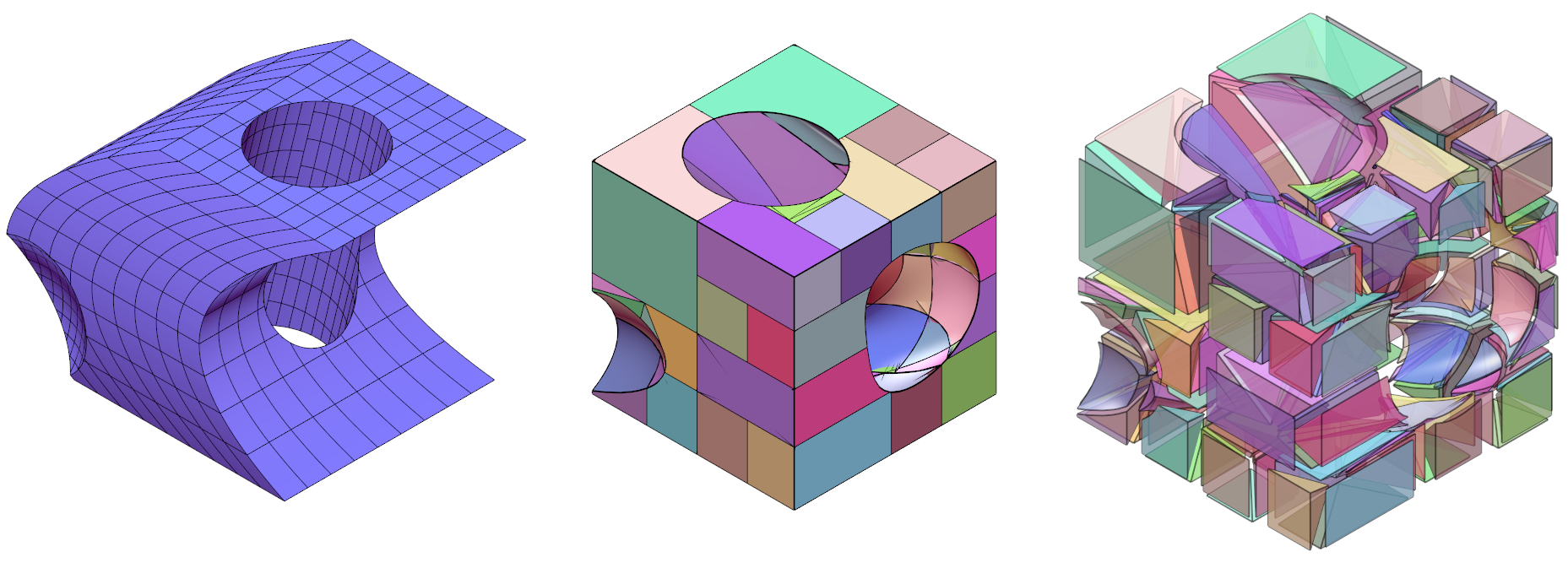} &
		\end{tabular}
	\end{center}
	\begin{picture}(0,0)
  	    \put(105, 30){$(a)$}
	    \put(250, 25){$(b)$}
	    \put(410, 25){$(c)$}
	\end{picture}
	\mbox{\vspace{-0.5in}}\\[-0.5in]
	\caption{Untrimming of a trimmed \Bezier{} trivariate. (a) A trimmed \Bezier{} trivariate of orders (2,2,3). (b) Untrimming result in the parametric space of (a), yielding 262 tensor product \Bspline{} trivariates shown in an exploded semi-transparent view in (c).}
	\label{fig:trivar2_untrim}
\end{figure}

In Figure~\ref{fig:solid3_untrim}, a V-rep model consisting of one trimmed trilinear trivariate is untrimmed in the parametric space, resulting in 76 tensor product trivariates that cover the trimmed parametric domain of the model. The trimming surfaces are untrimmed and approximated by bi-quadratic surfaces in the parametric domain. The untrimming algorithm took 52 seconds, out of which the approximation of the untrimmed surfaces in the parametric domain took 33 seconds.
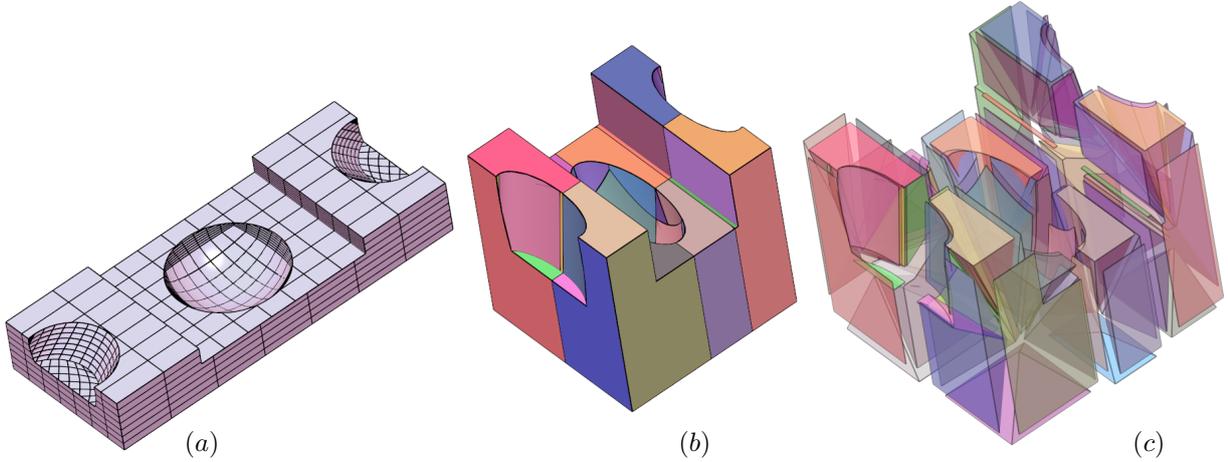
\begin{figure*}
	\begin{center}
		\begin{tabular}{cccc}
			\mbox{\hspace{-0.21in}}
			\includegraphics[scale = 0.33]{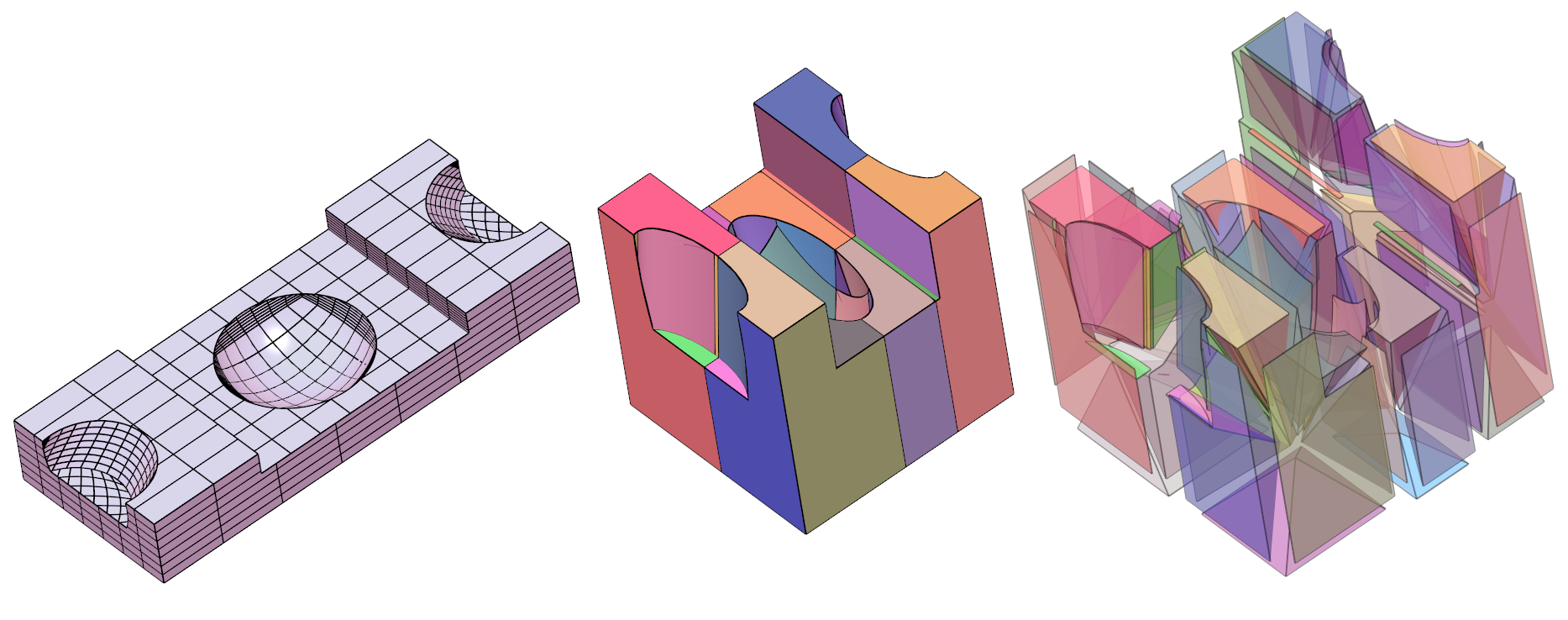} &
		\end{tabular}
	\end{center}
	\begin{picture}(0,0)
	\put( 65, 25){$(a)$}
	\put(250, 25){$(b)$}
	\put(420, 25){$(c)$}
	\end{picture}
	\mbox{\vspace{-0.5in}}\\[-0.5in]
	\caption{Untrimming of a V-model. (a) The trimmed trivariate of the V-model. Untrimming result consisting of 76 tensor product \Bspline{} trivariates composing the parametric domain in (b) and shown in semi-transparent exploded view in (c).}
	\label{fig:solid3_untrim}
\end{figure*}

Figure~\ref{fig:trimmed5_untrim} shows untrimming of a complex trimmed
trilinear \Bspline{} trivariate having (4,4,3) internal \Bezier{}
domains. The untrimming is performed in the Euclidean space, resulting
in 406 tensor product \Bspline{} trivariates. The trimmed trivariate
is first subdivided at all internal knots into 42 trimmed \Bezier{}
trivariates (and not $48=4\times4\times3$, as some domains are
completely trimmed away), and then each trimmed \Bezier{} trivariate
is untrimmed into a set of tensor product \Bspline{} trivariates. The
untrimming process took 10.5 seconds, out of which the subdivision of
the trimmed trivariate into trimmed \Bezier{} trivariates took 1.1
seconds.

\begin{figure}
\centering
    \begin{minipage}{.48\textwidth}
        \centering
	\begin{tabular}{c}
	    \mbox{\hspace{-0.21in}}
	    \includegraphics[scale = 0.14]{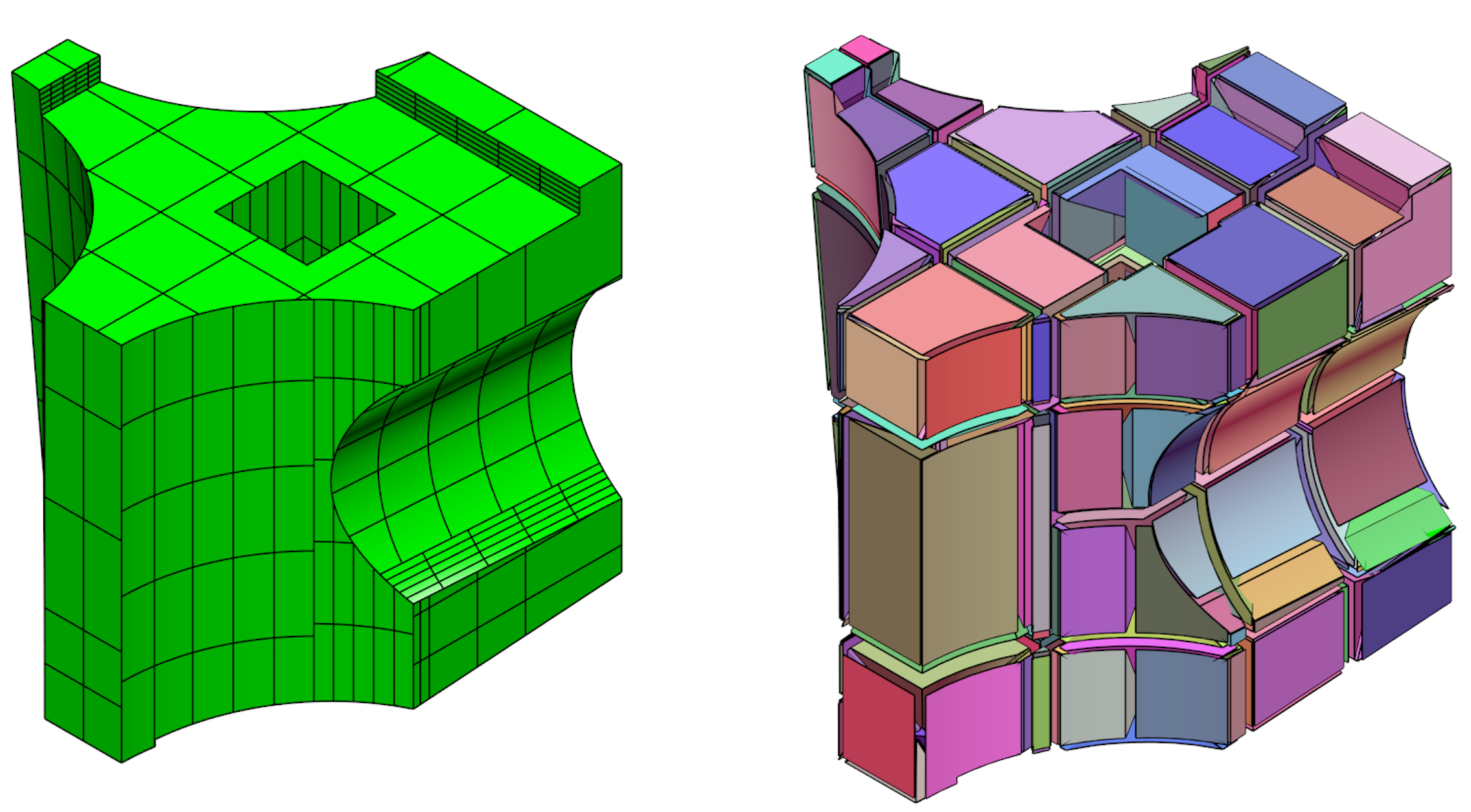}
	\end{tabular}
	\begin{picture}(0,0)
	    \put(-40, 15){$(a)$}
	    \put( 90, 15){$(b)$}
	\end{picture}
	\mbox{\vspace{-0.2in}}\\[-0.2in]
	\caption{Untrimming of a trilinear trimmed \Bspline{} trivariate with ($4\times4\times3$) \Bezier{} domains. (a) The trimmed \Bspline{} trilinear trivariate. (b) Untrimming of (a) yielding 406 tensor product \Bspline{} trivariates displayed in an exploded view.}
	\label{fig:trimmed5_untrim}
    \end{minipage}
    \mbox{\hspace{0.1in}}
    \begin{minipage}{.48\textwidth}
	\begin{tabular}{cc}
	    \mbox{\hspace{-0.21in}}
	    \includegraphics[scale = 0.18]{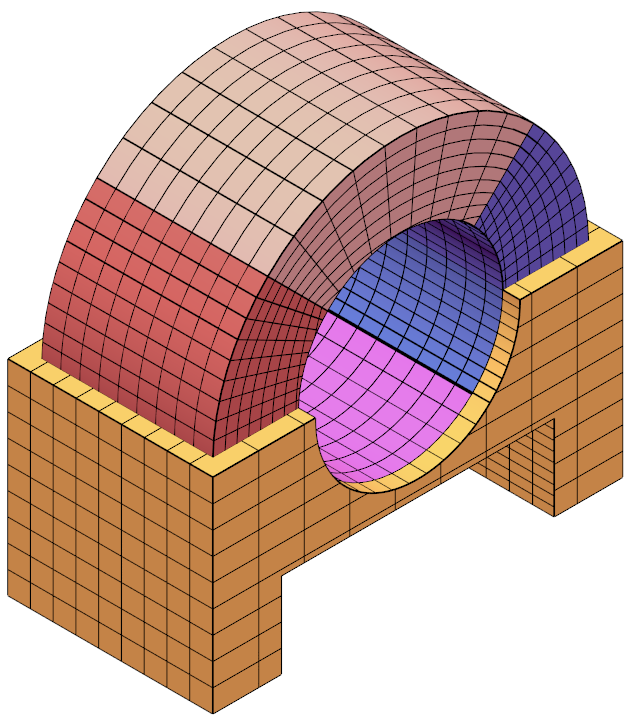} &
	    \includegraphics[scale = 0.18]{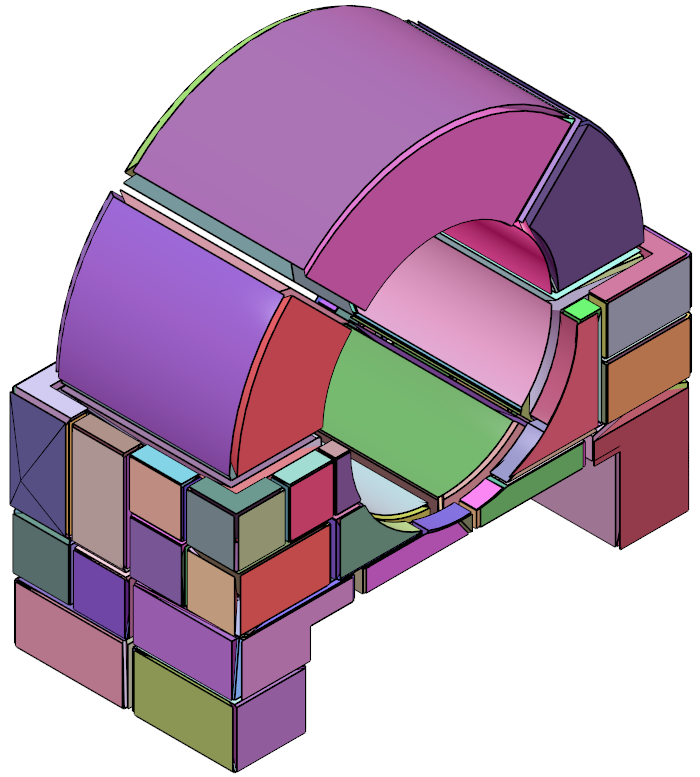}
	\end{tabular}
	\begin{picture}(0,0)
	    \put(60, 15){$(a)$}
	    \put( 190, 15){$(b)$}
	\end{picture}
	\mbox{\vspace{-0.2in}}\\[-0.2in]
	\caption{Untrimming of a V-model. (a) The V-model consisting of seven trimmed trivariates. (b) The result of untrimming each of the trimmed trivariates of the model, resulting in 294 tensor product \Bspline{} trivariates displayed in an exploded view.}
	\label{fig:vsolid8_untrim}
    \end{minipage}
\end{figure}

Figures~\ref{fig:vsolid8_untrim} and~\ref{fig:vsolid9_untrim} show
two untrimming results of volumetric models (V-models). Each V-model
consists of several trimmed trivariates. Since the trimmed trivariates
do not share the same parametric space, and in order to emphasize how
the untrimming result covers the entire V-model, the untrimming is
performed in the Euclidean space. In Figure~\ref{fig:vsolid8_untrim},
the V-model, consisting of seven mutually exclusive trimmed
trivariates, is untrimmed in 13.2 seconds, resulting in 294 tensor
product \Bspline{} trivariates. In Figure~\ref{fig:vsolid9_untrim},
another V-model, consisting of 20 mutually exclusive trimmed
trivariates, is untrimmed in 20.9 seconds, resulting in 370 tensor
product \Bspline{} trivariates.

\begin{figure}
\centering
    \begin{minipage}{.58\textwidth}
	\begin{tabular}{cc}
	    \mbox{\hspace{-0.15in}}
	    \includegraphics[scale = 0.36]{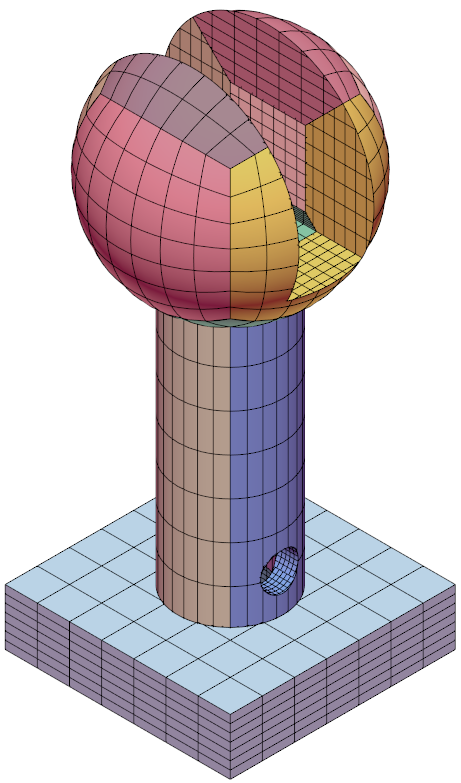} &
	    \includegraphics[scale = 0.36]{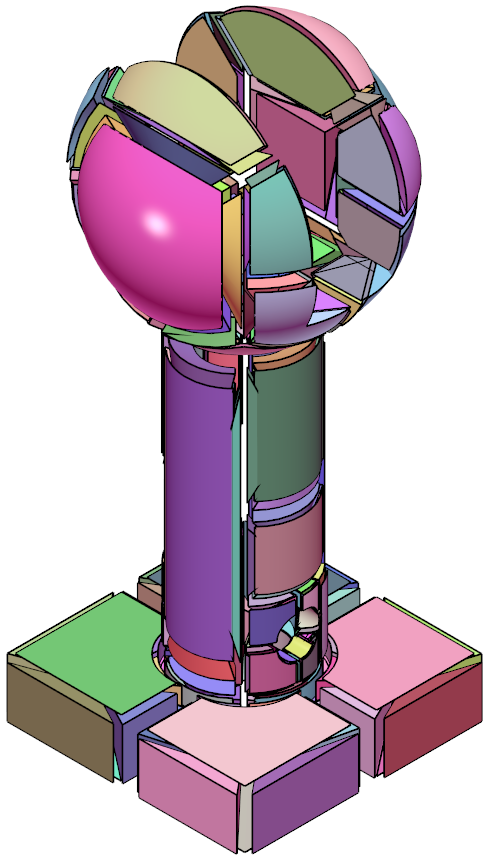}
	\end{tabular}
	\begin{picture}(0,0)
	    \put(105, 15){$(a)$}
	    \put(245, 15){$(b)$}
	\end{picture}
	\mbox{\vspace{-0.35in}}\\[-0.35in]
	\caption{Untrimming of a V-model. (a) The V-model consisting of 20 trimmed trivariates. (b) The result of untrimming each of the trimmed trivariates of the model, resulting in 370 tensor product \Bspline{} trivariates displayed in an exploded view.}
	\label{fig:vsolid9_untrim}
    \end{minipage}
    \mbox{\hspace{0.05in}}
    \begin{minipage}{.38\textwidth}
	\begin{tabular}{c}
	    \mbox{\hspace{-0.15in}}
	    \includegraphics[scale = 0.21]{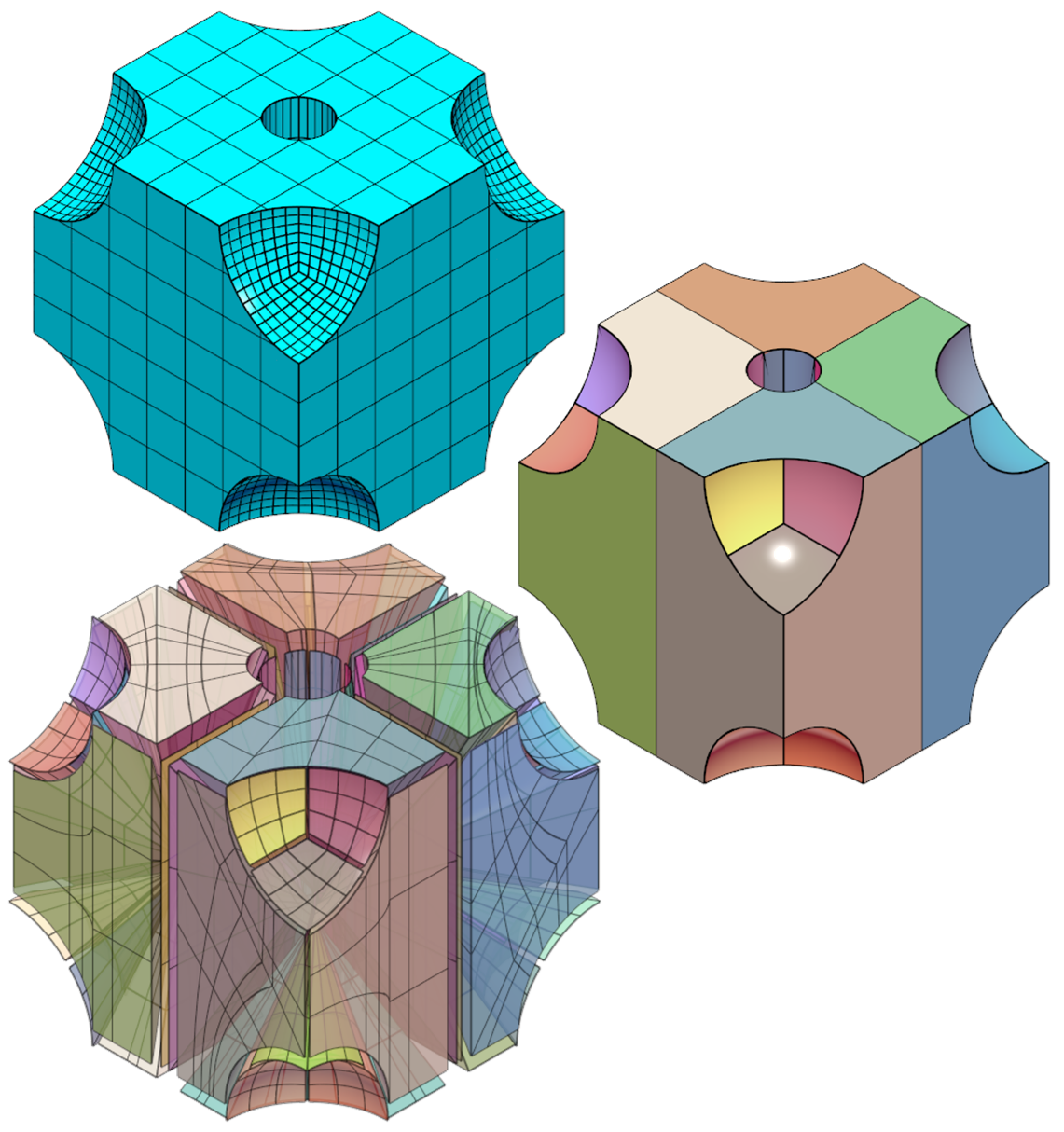}
	\end{tabular}
	\begin{picture}(0,0)
	    \put( -2, 120){$(a)$}
	    \put(155,  60){$(b)$}
	    \put( 90,  20){$(c)$}
	\end{picture}
	\mbox{\vspace{-0.3in}}\\[-0.3in]
	\caption{(a) A trimmed trivariate constructed by subtracting from a unit cube eight spheres of radius 0.3 centered at the eight corners of the cube, and a cylinder of radius 0.1 from the center of the cube. (b) Untrimming of (a) results in 56 tensor products that are displayed in an exploded view in (c).}
	\label{fig:volume_untrim}
    \end{minipage}
\end{figure}

Precise computation of integral properties over trimmed trivariates
are more difficult than over tensor products.  In
Figure~\ref{fig:volume_untrim}, we show one application of the
untrimming algorithm. We precisely compute the trimmed trivariate's
volume by first untrimming the trimmed trivariate into tensor
products, and then computing the volume of all covering tensor
products.  Since we use symbolic spline integration~\cite{Elber92},
the geometry must be (piecewise) polynomial (and not rational). Hence,
arcs and circles are approximated using piecewise polynomials to an
accuracy of ${\sim}10^{-3}$.  The object in
Figure~\ref{fig:volume_untrim}(a) is constructed by subtracting the
following from a unit cube: (i) eight spheres of radii 0.3 centered at
the corners of the cube, and (ii) a cylinder with a radius of 0.1
along the center of the unit cube.  The analytic value of the trimmed
cube's volume is 0.855486. While the value we compute is 0.855284, a
result that is well within the arc approximation.

Finally, some statistics on the untrimming process
(Algorithm~\ref{alg:BezierUntrimAlgo}) of the trimmed trivariates in
Figures~\ref{fig:trivar2_untrim}-\ref{fig:volume_untrim} are presented
in Table~\ref{table:vuntrimming_stats}. In the input part, the number
of trimmed trivariates in each model and the number of total trimming
surfaces for all the trimmed trivariates are presented in the first
and the second columns respectively.  In the output section, the first
column shows the number of tensor product trivariates in the result of the
untrimming algorithm. The second column shows the total number of
subdivisions of the domain performed during the untrimming
process. The third column shows the maximal depth of the recursive
subdivision calls, in the untrimming algorithm. The fourth column
shows the total running time, in seconds, and the last \NewT{column}
indicates if the untrimming is done in the parametric space or in the
Euclidean space.

\section{Analysis using untrimmed trivariates} \label{sec:analysis}
In this section, we illustrate the use
of the presented untrimming methodology for
performing isogeometric analysis in domains defined with trimmed B-spline trivariates. 
As a matter of example we focus in this section on a linear elasticity problem setting,
that is governed by the variational equation:
\begin{equation} \label{eq:elasticity}
\int_\mathcal{T}\mu(\bm{x})\nabla^s\bm{u}(\bm{x}):\nabla^s\bm{v}(\bm{x})\text{d}\bm{x}
+\int_\mathcal{T}\lambda(\bm{x})\nabla\cdot\bm{u}(\bm{x})\,\nabla\cdot\bm{v}(\bm{x})\text{d}\bm{x}
=\int_\mathcal{T}\bm{f}(\bm{x})\cdot\bm{v}(\bm{x})\text{d}\bm{x}
+\int_{\mathcal{S}_N}\bm{g}(\bm{x})\cdot\bm{v}(\bm{x})\text{d}\bm{x},
\end{equation}
where $\bm{u}:\mathcal{T}\to\mathbb{R}^3$ is the elastic displacement at every point of the domain,
the problem unknown, and $\bm{v}:\mathcal{T}\to\mathbb{R}^3$ corresponds to the test functions.
On the other hand $\bm{f}:\mathcal{T}\to\mathbb{R}^3$ and 
$\bm{g}:\mathcal{S}_N\to\mathbb{R}^3$ are the volumetric forces (e.g.\ self-weight) and the external loads applied
on the external boundary (e.g.\ pressure). $\mathcal{S}_N$ is a subset of the exterior boundary $\mathcal{S}$ in which
external loads are applied.
Finally, $\lambda$ and $\mu$ are the Lam\'e parameters that charaterize the behaviour of the elastic material,
that may change from point to point.
For the sake of brevity, prescribed displacements (Dirichlet boundary conditions) are not discussed here.
In an IGA context, both $\bm{u}$ and $\bm{v}$ are discretized by means of trivariate B-splines:
\begin{equation}
	\bm{u}=\sum_{i=0}^{l}\sum_{j=0}^{m}\sum_{k=0}^{n}
	\bm{u}_{i,j,k} B_{i,d_{u}}(u) B_{j,d_{v}}(v) B_{k,d_{w}}(w),
\end{equation}
where the coefficients $\bm{u}_{i,j,k}\in\mathbb{R}^3$ are the problem unknowns ($\bm{v}$ is discretized in the same way).
We refer the interested readers to \cite{iga} for a more detailed discussion about the fundamentals of IGA.
\NewT{Thus, in an isoparametric framework we use the same spline space for describing both the trivariate $T$ (as in Equation \eqref{eqn:triv}) and discretizing the solution (trial) $\bm{u}$ and the test functions $\bm{v}$.
For the analysis, the support of functions $B_{i,d_{u}}$, $B_{j,d_{v}}$, $B_{k,d_{w}}$ is limited to the active region of the domain defined by the trimming. Therefore, those functions whose support is completely outside of the trimmed domain will not be considered in the analysis.}

In order to compute the integrals present in the Equation \eqref{eq:elasticity}
in the domain $\mathcal{T}$ (and on the boundary $\mathcal{S}_N$)
we decompose the integral over the full domain as the sum of the integrals
in every single \Bezier{} element contained in the domain.
Thus, for computing the integral of a generic quantity $\alpha(\bm{x})$ over
the domain $\mathcal{T}$ we split the integral as:
$\int_\mathcal{T}\alpha(\bm{x})\text{d}\bm{x}=\sum^{n_b}_{i}\int_{\hat{\mathcal{B}_i}}\hat\alpha(\bm{x})\text{d}\bm{x}$,
where $\hat{\mathcal{B}_{i}}$ are the representation, in the parametric domain of $T$,
of the trimmed \Bezier{} trivariates $\mathcal{B}_{i}$ (such that $\mathcal{T}=\cup^{n_b}_{i}\mathcal{B}_{i}$),
and $n_b$ is the number of trivariates.
$\hat\alpha(\bm{x})$ corresponds the pull-back of $\alpha(\bm{x})$.
Thus, when computing $\int_{\hat{\mathcal{B}_i}}\hat\alpha(\bm{x})\text{d}\bm{x}$ for $i=1,\dots,n_b$
there exist two possible situations:
\begin{itemize}
\item If $\hat{\mathcal{B}_i}$ is not a trimmed trivariate, but a full one, then
a \NewT{Gaussian} quadrature rule is applied for computing the integral.
This is the case of standard IGA methods for non-trimmed domains.
\item Otherwise, if $\hat{\mathcal{B}_i}$ is a trimmed trivariate, the element is untrimmed
according to methodology presented in Section \ref{sec:algoSec}, and the integral
is computed using the resulting untrimming trivariates as:
\begin{equation}
\int_{\hat{\mathcal{B}_i}}\hat\alpha(\bm{x})\text{d}\bm{x} =
\sum^{n_{u,i}}_{j=1}\int_{\tau_{i,j}}\hat\alpha(\bm{x})\text{d}\bm{x},\qquad\textrm{for }i=1,\dots,n_b,
\end{equation}
where $\tau_{i,j}$ is the $j$-th untrimming tensor product trivariate of the \Bezier{} element $\hat{\mathcal{B}_i}$,
with $j=1,\dots,n_{u,i}$.
\end{itemize}
A discussion regarding the integration of the trimmed and non-trimmed \Bezier{} elements can be found 
in \cite{Rank2012}. It is important to remark here that the tensor product trivariates $\tau_{i,j}$ 
are only used for integration purposes. Therefore, they are not required
to have high-quality \NewT{Jacobians}, and it is sufficient that they are singular only on their boundary.

\begin{remark}
To perform analyses, in an IGA framework, in the case of computational domains created as the union of more than one trimmed B-spline trivariate
(e.g.\ Figures~\ref{fig:vsolid8_untrim} and \ref{fig:vsolid9_untrim})
is not as straightforward as for single trimmed trivariates.
These situations involve not only the precise integration of the operators described in Equation~\eqref{eq:elasticity},
but also the consistent gluing of all the partial solution discretizations, defined for every single trivariate, in the intersection regions.
The treatment of this kind of domains is out of the scope of this paper.
\end{remark}

In order to illustrate the potentiality of the presented procedure, we perform linear elasticity numerical
experiments for the trimmed geometries shown in Figures \ref{fig:trivar2_untrim}, \ref{fig:trimmed5_untrim} and \ref{fig:volume_untrim}.

The same elastic material is considered in all cases, being the Young modulus and the Poisson ratio,
$E=1\,\text{MPa}$ and $\nu=0.3$, respectively (where $\lambda=E\nu/(1+\nu)(1-2\nu)$ and $\mu=E/2(1+\nu)$).
The analysis were performed using the IGA library igatools described in \cite{Igatools2015}.

\NewT{All three cases present analogous loading conditions:
one face is completely fixed (no displacement in any direction is allowed)
while the opposite face is pulled perpendicularly (the pulled face is free to deform transversally). In these cases, Dirichlet boundary conditions are applied on faces of $T$ (that could potentially be trimmed) in a strong way:
the degrees of freedom associated to the basis functions whose traces have support on those faces are prescribed. On the other hand, the imposition of Dirichlet boundary conditions on trimming boundaries (that are not faces of $T$)
requires the use of weak imposition methods. This is out of the scope of this paper, however, we refer the interested reader to \cite{ruess_weakly_2013} and \cite{buffa_minimal_2019}, and references therein.}

In Figure \ref{fig:analysis}, we gathered some of the obtained results.
\begin{figure}
\centering
	\begin{tabular}{ccc}
	    %\mbox{\hspace{-0.15in}}
	    \includegraphics[width=0.32\textwidth]{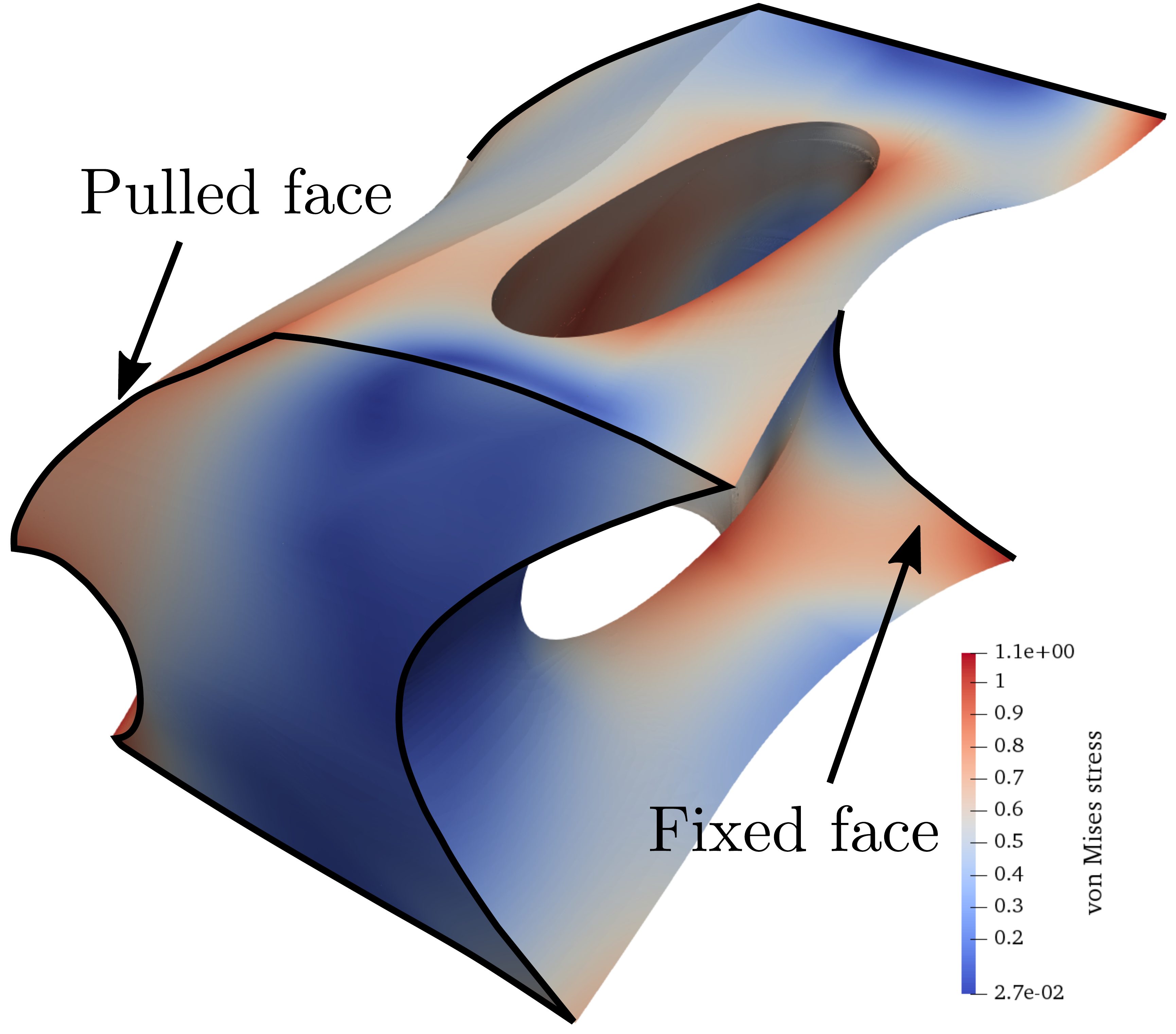} &
	    \includegraphics[width=0.27\textwidth]{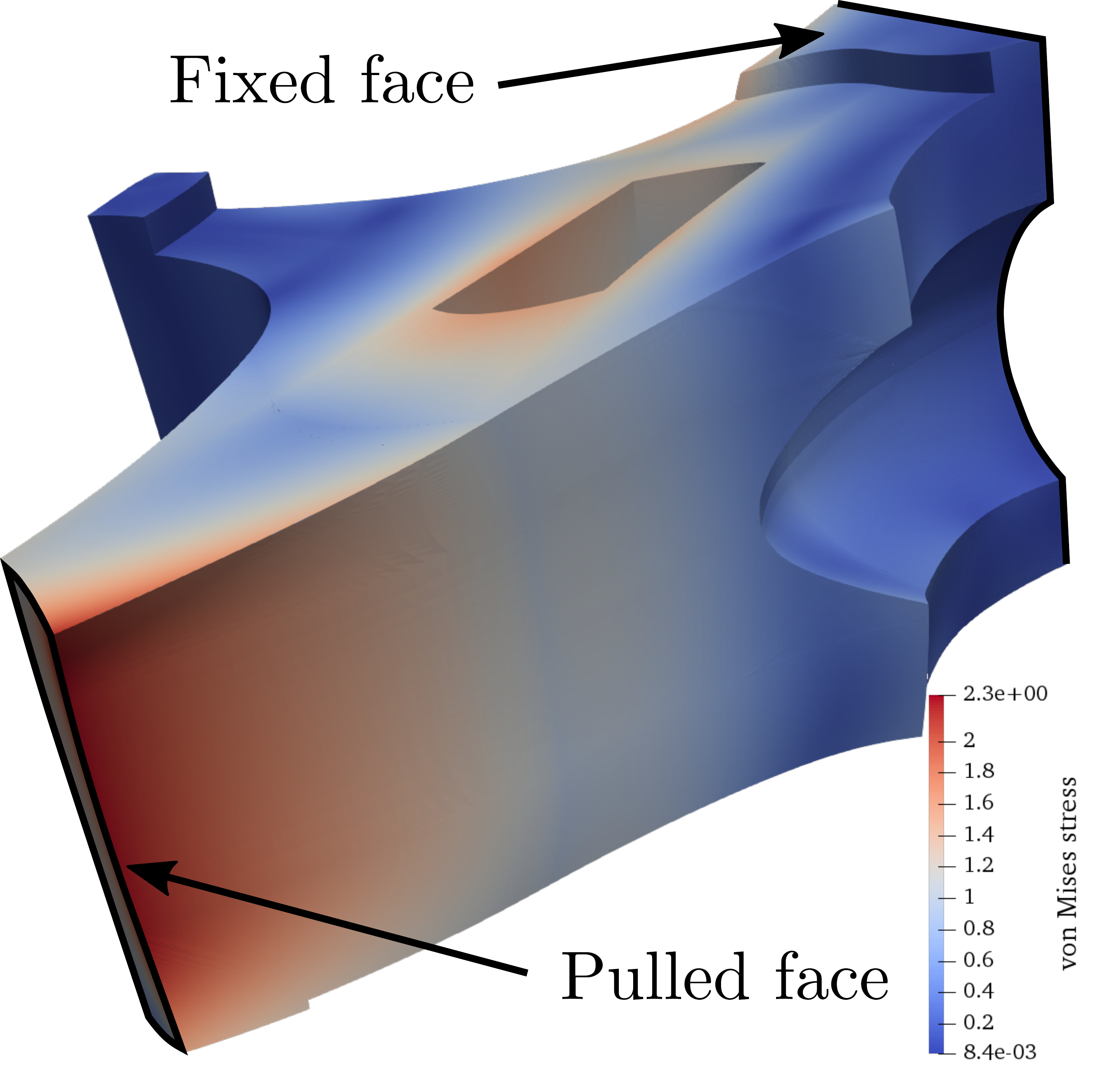} &
	    \includegraphics[width=0.34\textwidth]{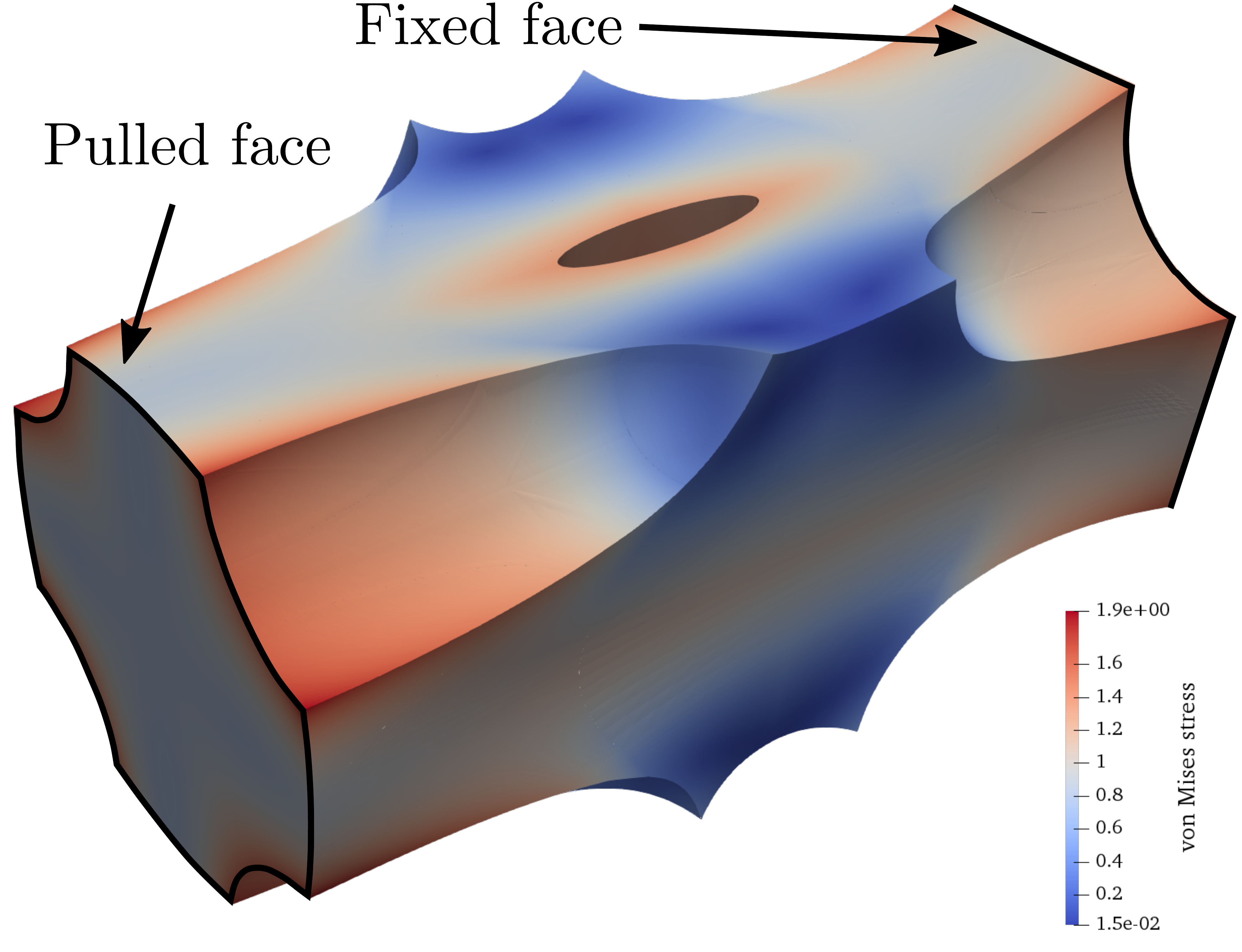}
	\end{tabular}
	\begin{picture}(0,0)
	    \put(-220, 15){$(a)$}
	    \put(-70, 15){$(b)$}
	    \put(70, 15){$(c)$}
	\end{picture}
 \caption{Linear elasticity analyses on trimmed geometries from Figures~
\ref{fig:trivar2_untrim},
\ref{fig:trimmed5_untrim} and \ref{fig:volume_untrim}.
 The von Mises stress distribution is plotted on the deformed geometries submitted to external loads: one face is fixed while
the opposite face is pulled perpendicularly.}
 \label{fig:analysis}
\end{figure}
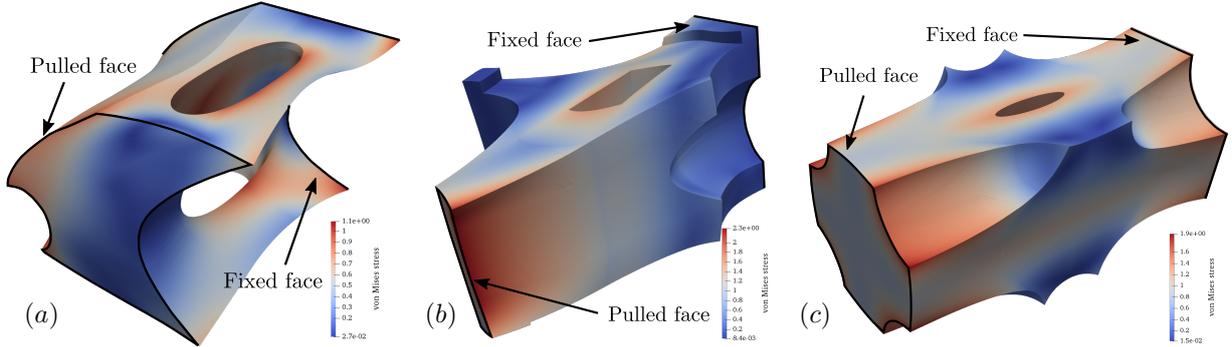
The shown geometries are deformed with respect to the original ones as a consequence of the applied loads.
The plotted (colored) scalar field corresponds the distribution of the von Mises stress.

The cases \ref{fig:analysis}(a), \ref{fig:analysis}(b) and  \ref{fig:analysis}(c) were obtained using trimmed B-spline trivariates
with $(3\times3\times3)$, $(4\times4\times3)$ and $(4\times4\times4)$ \Bezier{} domains, respectively,
and discretizing the elastic displacement solutions with degree $2$ in every direction
(i.e.\ the models have $375$, $540$ and $648$ degrees of freedom, respectively).
Note that the models in Figures~\ref{fig:trivar2_untrim} and~\ref{fig:volume_untrim} have underlying \Bezier{} trivariates, and in order to obtain a \NewT{more precise} analysis solution, the trivariates of these models were refined in each parametric direction.
The computational times for the assembly of the stiffness matrices are $114$, $32$ and $61$ seconds, respectively;
and the resolution of the linear system of equations took less than one second in all cases.
The current untrimming process creates very accurate reparametrization of the trimmed domains.
Nevertheless, from the analysis point of view such level of accuracy is not needed for integration purposes.
Therefore, by creating much coarser reparametrizations the computational time of the matrix assembly
could potentially be reduced by orders of magnitude.

\section{Conclusion}	\label{sec:concludeSec}

In this work, an untrimming algorithm for trimmed trivariate is
introduced: decomposing a trimmed \Bspline{} trivariate into a set of
mutually exclusive tensor product \Bspline{} trivariates that
completely cover the trimmed domain. The algorithm uses a subdivision
algorithm, introduced in this paper, that precisely subdivides the
trimmed \Bspline{} trivariate into set of trimmed \Bezier{}
trivariates. The untrimming algorithm then generates tensor product
\Bspline{} trivariates that are singular only on the boundaries and
thus can be utilized for integration in IGA application (i.e at
quadrature locations). \NewT{The quality of the analysis solution is
not influenced by the Jacobians' quality of the generated tensor
product trivariates.}

Several directions for further improvements can be sought. Additional strategies of selecting the subdivision location, in case no kernel point is found (see Section~\ref{sec:No_Kernel_Subd}), can be further investigated, in order to minimize the number of subdivisions required and thus minimizing the number of generated trivariates in the output. For example, subdivision at locations that isolate individual non isoparametric trimming surfaces, if exist. That is, every trimmed trivariate will have at most one non iso-parametric trimming surface.

\NewT{Aiming for less subdivisions when seeking the kernel points
(see Section~\ref{sec:Kernel_Point}), one can explore better sampling
approaches than the uniform sampling approach taken here.  Similarly,
a more precise (less lossy) approach can be sought, by computing a
tight bounding box of the solutions of Equation~\eqref{eqn:visibility}
in Lemma~\ref{lemma:vis}, for all trimming surfaces,
simultaneously. However, it requires the (simultaneous) processing of
multiple bivariate inequalities (one for each trimming surface).
While potentially possible, for example, using interval arithmetic, it
can be time consuming.}

Due to the independent untrimming of adjacent boundary surfaces,
adjacent tensor product trivariates in the output might not share the
same functional space, which makes the presented approach less
suitable for other analysis approaches, such as domain
decomposition~\cite{domain_dec_widlund}, where constraints are imposed
on shared boundaries. In order to have adjacent surfaces with the same
functional space, the (trimmed) surfaces untrimming process (see
Section~\ref{sec:Untrimming_surfaces}) needs to be adapted to consider
not only a single trimmed surface, but also the adjacencies between
the trimming (trimmed) surfaces of the trivariate.

The back projection approximation process of trimming surfaces, in the parametric space, as described in Section~\ref{sec:ParamSpaceApprox}, 
doesn't guarantee stitched boundaries between adjacent surfaces, which likely to lead to black holes in the approximated surfaces that may introduce inaccuracies when using the untrimming results in applications, such as analysis. To overcome this limitation, the topological adjacency information should be provided to the back projection process, and proper stitching methods should be considered.

\section*{Acknowledgements}
This research was supported in part with funding from the ISRAEL SCIENCE FOUNDATION (grant No.\ 597/18) and in part the Defense Advanced Research Projects Agency (DARPA), under contract HR0011-17-2-0028. The views, opinions and/or findings expressed are those of the author and should not be interpreted as representing the official views or policies of the Department of Defense or the U.S.\ Government.
Pablo~Antolin gratefully acknowledges the support of the European Research Council, through the ERC AdG n.\ 694515 - CHANGE.

\bibliography{VUntrimBib}

\end{document}